\documentclass[12pt]{article}
\usepackage{LATE_style}

\defcitealias{AI1995}{AI}

\shortcites{CSHNW2012}
\shortcites{DHMMR2019}
\shortcites{Finkelsteinetal2012}

\shortcites{Acemoglu2008}
\shortcites{Autor2013}
\shortcites{Becker2011}
\shortcites{BCFH2017}
\shortcites{ChodorowReich2012}
\shortcites{Chou2010}
\shortcites{Draca2011}
\shortcites{Lipscomb2013}
\shortcites{Moser2014}

\begin{document}
\singlespacing

\title{\textsc{\Large{When Should We (Not) Interpret Linear IV Estimands as LATE?}}\thanks{First arXiv draft: November 13, 2020. This version: April 2, 2026. I am very grateful to the Editor, three anonymous referees, Martin Andresen, Isaiah Andrews, Carolina Caetano, Brant Callaway, Todd Elder, Brigham Frandsen, Xavier D'Haultfœuille, Molly Hawkins, Phillip Heiler, Stefan Hoderlein, Felix Holub, Martin Huber, Peter Hull, Shoya Ishimaru, Krzysztof Karbownik, Toru Kitagawa, Patrick Kline, Michael Knaus, Michal Kolesár, Giovanni Mellace, Anna Mikusheva, Áureo de Paula, Alex Poirier, Jonathan Roth, Pedro Sant'Anna, Vira Semenova, Jesse Shapiro, Gary Solon, Liyang Sun, Alex Torgovitsky, Edward Vytlacil, Chris Walters, Jeff Wooldridge, Kaspar Wüthrich, Yinchu Zhu, seminar participants at Binghamton University, Brandeis University, Duke University, Emory University, Goethe University Frankfurt, University of Chicago, and University of Warsaw, and participants at many conferences for helpful comments and discussions. I thank Megan Stevenson for sharing data and code, Tiemen Woutersen for sharing code, Anne Laski and Qihui Lei for excellent research assistance, as well as Danny Hellstein for assistance with the Brandeis HPC cluster. Qihui Lei has also coauthored the companion MATLAB, R, and Stata packages \texttt{fejiv}. Finally, I acknowledge financial support from the Theodore and Jane Norman Fund.}
}
\author{\textsc{\Large{Tymon Słoczyński}}\thanks{Brandeis University. Correspondence:~Department of Economics, Brandeis University, MS 021, 415 South Street, Waltham, MA 02453. E-mail:~tslocz@brandeis.edu.}
}
\date{}

\begin{titlepage}

\maketitle

\begin{abstract}
\noindent
In this paper I revisit the interpretation of the linear instrumental variables (IV) estimand as a weighted average of conditional local average treatment effects (LATEs). I~focus on a situation in which additional covariates are required for identification while the reduced-form and first-stage regressions may be misspecified due to an implicit homogeneity restriction on the effects of the instrument. I show that the weights on some conditional LATEs are negative and the IV estimand is no longer interpretable as a causal effect under a weaker version of monotonicity, \textit{i.e.}~when there are compliers but no defiers at some covariate values and defiers but no compliers elsewhere. The problem of negative weights disappears in the interacted specification of \cite{AI1995}, which avoids misspecification and seems to be underused in applied work. I illustrate my findings in an application to the causal effects of pretrial detention on case outcomes. In this setting, I reject the stronger version of monotonicity, demonstrate that the interacted instruments are sufficiently strong for consistent estimation using the jackknife methodology, and present several estimates that are economically and statistically different, depending on whether the interacted instruments are used.
\end{abstract}



\thispagestyle{empty}

\end{titlepage}

\setcounter{page}{2}

\onehalfspacing

\section{Introduction}

Many instrumental variables are only valid after conditioning on additional covariates. The draft eligibility instrument in \cite{Angrist1990} requires controlling for the year of birth. The college proximity instrument in \cite{Card1995} is invalid without conditioning on several individual characteristics of workers \citep{Kitagawa2015}. Even in the case of randomized experiments with noncompliance, it is often necessary to control for covariates correlated with treatment probability, such as household size and survey wave in \cite{Finkelsteinetal2012}.

When conditioning on additional covariates is necessary for instrument validity, interpreting the linear instrumental variables (IV) and two-stage least squares (2SLS) estimands becomes complicated. \cite{AI1995} (hereafter, \citetalias{AI1995}) provide an influential interpretation of the 2SLS estimand in this context as a convex combination of conditional local average treatment effects (LATEs), \textit{i.e.}~average effects of treatment for individuals whose treatment status is affected by the instrument. However, this result is restricted to saturated models with discrete covariates and first-stage regressions that include a complete set of interactions between these covariates and the instrument. Such specifications are rare in empirical work, as is evident from several recent surveys of applications of IV methods.\footnote{\cite{BBMT2022,BBMT2025} consider a sample of 99 papers and find a single application of \citetalias{AI1995}'s specification. \cite{MTW2021} consider a sample of 122 papers and identify seven with specifications that include \textit{some} covariate interactions with a single instrument.} This makes \citetalias{AI1995}'s result inappropriate for interpreting the vast majority of IV estimates encountered in economic applications \citep[cf.][]{Abadie2003}.

In this paper I revisit the question of the causal interpretability of standard instrumental variables estimands. In particular, I focus on whether these estimands can be written as weighted averages of conditional LATEs with positive weights and, if so, whether these weights have an intuitive interpretation. To do so, I consider two variants of the usual monotonicity assumption: ``weak monotonicity,'' which postulates that at every covariate value, the instrument either does not discourage or does not encourage anyone to take treatment, and ``strong monotonicity,'' which additionally requires that the direction of this effect is uniform across covariate values.

My first contribution is to demonstrate that under weak monotonicity, the weights on some conditional LATEs may be negative in the usual application of IV, which restricts the first-stage effects of the instrument to be homogeneous. This finding implies that the resulting estimand is not a useful summary measure of average treatment effects; this parameter could be negative (positive) even if treatment effects are positive (negative) for everyone in the population. Under the same assumptions, all weights are necessarily positive in \citetalias{AI1995}'s interacted specification.

My second contribution is to explicitly compare the weights in both specifications with the ``desired'' weights, which recover the unconditional LATE parameter. Under strong monotonicity, when the weights in the usual application of IV and \citetalias{AI1995}'s specification are positive, both specifications overweight the effects in groups with large variances of the instrument, while the latter also overweights the effects in groups with strong first stages. It follows that the usual application of IV might be preferable when violations of strong monotonicity are not an issue.

However, if weak monotonicity is plausible but strong monotonicity is not, my theoretical results suggest that \citetalias{AI1995}'s interacted specification is preferable to the usual application of IV\@. Unfortunately, \citetalias{AI1995}'s specification is also difficult to estimate without bias; when the researcher divides the sample into many groups and subsequently creates an interacted instrument for each, 2SLS will be subject to the ``many instrument'' bias (see, \textit{e.g.}, \citealp{Bekker1994}). An alternative approach to estimating \citetalias{AI1995}'s specification, such as the fixed effect jackknife IV (FEJIV) estimator of \cite{CSW2023}, should be used instead. Another concern about specifications with many instruments is whether they are jointly strong enough to enable consistent estimation. In this context, I consider a recent pretest for weak identification developed by \cite{MS2022}. As an illustration, I perform an extensive simulation study. In these simulations, \cite{MS2022}'s pretest does a great job differentiating between cases where the best estimators of \citetalias{AI1995}'s specification, such as FEJIV, perform well and cases where all estimators perform badly.

To corroborate the concern about violations of strong monotonicity, I also replicate a sample of 988 instrumental variables regressions from 25 papers published in journals of the American Economic Association between 2006 and 2015. Every specification in my sample is based on a linear first-stage regression that restricts the effects of the instrument to be homogeneous. If strong monotonicity is violated but weak monotonicity is not, the homogeneous first stage will be misspecified and the conditional first stage will be positive for some covariate values but negative for others. First, I present strong suggestive evidence of the latter phenomenon, which directly translates to the incidence of negative weights in the usual application of IV\@. Then, I formally reject the null hypothesis of first-stage homogeneity in more than 70\% of specifications in an average paper, despite accounting for multiple hypothesis testing.

In addition, I illustrate my findings in an application to the causal effects of pretrial detention on case outcomes \citep{Stevenson2018}. Here, I consider several saturated specifications, which allows me to compare the estimates of \citetalias{AI1995}'s specification with the usual application of IV\@. I can also formally test whether the conditional first stage is positive for some covariate values and negative for others, and I conclusively reject the null hypothesis of sign homogeneity. The estimates based on \citetalias{AI1995}'s specification are smaller than in the usual application of IV, and the difference is often statistically significant. \cite{MS2022}'s pretest rejects in every case I consider, which supports the notion that the estimates based on \citetalias{AI1995}'s specification are preferable.

Finally, I supplement this paper with companion MATLAB, R, and Stata packages, \texttt{fejiv}, available at the MATLAB Central File Exchange, the Comprehensive R Archive Network (CRAN), and the Statistical Software Components (SSC) Archive, respectively.\footnote{To download the MATLAB package from the File Exchange, search for \texttt{fejiv} in the Add-On Explorer. To download the R package from CRAN, type \,\texttt{install.packages("fejiv")}\, in the R/RStudio console. To download the Stata package from SSC, type \,\texttt{ssc install fejiv}\, in the Command window.} These packages, based on the MATLAB code of \cite{CSW2023}, can be used to implement the FEJIV estimator in practice. Of these, the MATLAB package is most appropriate when working with large datasets. A Stata package to implement \cite{MS2022}'s pretest is also available from \cite{Sun2023}.

Two papers closely related to this are \cite{Kolesar2013} and \cite{BBMT2022,BBMT2025}. Like this paper, \cite{Kolesar2013} studies the interpretation of 2SLS estimands under weak monotonicity while also considering the probability limits of several jackknife-type estimators, limited information maximum likelihood (LIML), and other alternatives to 2SLS\@. \cite{Kolesar2013}'s main result on 2SLS is not particular to any specification but instead represents a generic two-step IV estimand as a weighted average of conditional LATEs. The resulting weights are positive, subject to an additional condition that needs to be verified on a case-by-case basis.\footnote{This condition essentially requires that the first stage postulated by the researcher provides a sufficiently good approximation to the true first stage \citep[cf.][]{HV2005}.} In contrast, this paper focuses specifically on the usual application of IV and \citetalias{AI1995}'s interacted specification. The benefit is that this allows me to considerably simplify the representation and obtain results that are more transparent and easy to interpret. This includes the novel result that in the usual application of IV, the weights on some conditional LATEs may be negative under weak monotonicity. In another contribution, released after this paper first circulated, \cite{BBMT2022,BBMT2025} focus on the consequences of misspecification of the model for the instrument propensity score that is implicit in IV and 2SLS estimation. In this paper I focus instead on violations of strong monotonicity and their implications.

The remainder of the paper is organized as follows. Section \ref{sec:framework} introduces my framework. Section \ref{sec:theory} provides my theoretical contributions, a review of the literature on many instruments, and a simulation study. Section \ref{sec:empirics} studies negative first stages and first-stage heterogeneity in a sample of recent applications of IV methods and illustrates my findings in an analysis of the causal effects of pretrial detention on case outcomes. Section \ref{sec:conclusion} concludes. The appendix contains my proofs as well as additional simulation and estimation results.

\section{Framework}
\label{sec:framework}

In this section I formally define the objects of interest, \textit{i.e.}~the conditional and unconditional IV and 2SLS estimands. I reserve the term ``2SLS'' for the appropriate estimand in a model with interacted instruments; see equation (\ref{eq:estimand_2sls}) below. When a single instrument is used instead, I use the term ``IV'' or ``linear IV''; see equation (\ref{eq:estimand_iv}). In what follows, I also review identification in the LATE framework with covariates \citep[cf.][]{Abadie2003}. Throughout the paper I assume that the appropriate moments exist whenever necessary.

\subsection{Notation and Estimands}
\label{sec:notation}

Suppose we are interested in the causal effects of a treatment, $D \in \{ 0,1 \}$, on an outcome, $Y = Y(D)$, where $Y(1)$ and $Y(0)$ are potential outcomes. An instrument, $Z \in \{ 0,1 \}$, is also available, and it determines which of the potential treatment states, $D(1)$ and $D(0)$, is observed, $D = D(Z)$. In principle, we could let $Y = Y(Z,D)$, but we will rule out direct effects of $Z$ on $Y$ below. Finally, let $X = \left( 1, X_1, \ldots, X_J \right)$ denote a row vector of covariates. In some cases I will allow for the possibility that additional instruments have been created by interacting $Z$ with all elements of $X$; then, $Z_{\mathrm{C}} = \left( Z, Z X_1, \ldots, Z X_J \right)$ will be used to denote the resulting row vector of instruments.

To provide motivation for what follows, let us consider the standard single-equation linear model for $Y$:
\begin{equation}
\label{eq:structural}
Y = D \beta + X \rho + \upsilon,
\end{equation}
where $X$ and the instrument(s) are assumed to be uncorrelated with the error term $\upsilon$. Also, $\beta$ is the coefficient of interest. In this paper I do not assume that equation (\ref{eq:structural}) is correctly specified; in particular, I allow the effect of $D$ on $Y$ to be correlated with both observables and unobservables.

In practice, however, many researchers act as if this model is correctly specified and use linear IV or 2SLS for estimation. In what follows, I will focus on the interpretation of the probability limits of the IV and 2SLS estimators of $\beta$ when equation (\ref{eq:structural}) is possibly misspecified. With a single instrument, the probability limit of linear IV or, simply, the (linear) IV estimand is
\begin{equation}
\label{eq:estimand_iv}
\beta_{\mathrm{IV}} = \left[ \left( \e \left[ Q^{\prime} W \right] \right) ^{-1} \e \left[ Q^{\prime} Y \right] \right] _1,
\end{equation}
where $W = \left( D, X \right)$, $Q = \left( Z, X \right)$, and $\left[ \cdot \right] _k$ denotes the $k$th element of the corresponding vector. Clearly, when a single instrument is available, equation (\ref{eq:estimand_iv}) characterizes the target of estimation in most empirical studies, which I also call the ``usual'' or ``standard'' estimand. This specification corresponds to reduced-form and first-stage regressions that project $Y$ and $D$ on $X$ and $Z$, excluding any interactions between $X$ and $Z$\@. Hence, I also refer to this specification as ``noninteracted.''

On the other hand, if a vector of interacted instruments, $Z_{\mathrm{C}}$, is used in 2SLS estimation of equation (\ref{eq:structural}), the relevant probability limit or, simply, the 2SLS estimand is
\begin{equation}
\label{eq:estimand_2sls}
\beta_{\mathrm{2SLS}} = \left[ \left( \e \left[ W^{\prime} Q_{\mathrm{C}} \right] \left( \e \left[ Q_{\mathrm{C}}^{\prime} Q_{\mathrm{C}} \right] \right) ^{-1} \e \left[ Q_{\mathrm{C}}^{\prime} W \right] \right) ^{-1} \e \left[ W^{\prime} Q_{\mathrm{C}} \right] \left( \e \left[ Q_{\mathrm{C}}^{\prime} Q_{\mathrm{C}} \right] \right) ^{-1} \e \left[ Q_{\mathrm{C}}^{\prime} Y \right] \right] _1,
\end{equation}
where $Q_{\mathrm{C}} = \left( Z_{\mathrm{C}}, X \right)$. In this specification, the corresponding reduced-form and first-stage regressions project $Y$ and $D$ on $X$ and $Z_{\mathrm{C}}$, which implies that the effects of $Z$ on $Y$ and $D$ are allowed to vary with $X$ due to the interactions between $X$ and $Z$\@. Thus, I also refer to this specification as ``interacted'' or ``fully interacted.''

Regardless of the implicit restrictions on the effects of the instrument, the true first stage can be written as
\begin{equation}
\label{eq:fs}
\e \left[ D \mid X, Z \right] = \psi(X) + \omega(X) \cdot Z,
\end{equation}
where
\begin{equation}
\label{eq:cfs}
\omega(x) = \e \left[ D \mid Z=1, X=x \right] - \e \left[ D \mid Z=0, X=x \right]
\end{equation}
is the conditional first-stage slope coefficient or, equivalently, the coefficient on $Z$ in the regression of $D$ on $1$ and $Z$ in the subpopulation with $X=x$. Similarly, the conditional IV (or Wald) estimand can be written as
\begin{equation}
\label{eq:wald}
\beta(x) = \frac{\e \left[ Y \mid Z=1, X=x \right] - \e \left[ Y \mid Z=0, X=x \right]}{\e \left[ D \mid Z=1, X=x \right] - \e \left[ D \mid Z=0, X=x \right]}.
\end{equation}
This parameter is equivalent to the coefficient on $D$ in the IV regression of $Y$ on $1$ and $D$ in the subpopulation with $X=x$, with $Z$ as the instrument for $D$\@.

\subsection{Local Average Treatment Effects}
\label{sec:late}

In the LATE framework of \cite{IA1994} and \cite{AIR1996}, the population consists of four latent groups: always-takers, for whom $D(1) = D(0) = 1$; never-takers, for whom $D(1) = D(0) = 0$; compliers, for whom $D(1) = 1$ and $D(0) = 0$; and defiers, for whom $D(1) = 0$ and $D(0) = 1$. As demonstrated by \cite{IA1994}, if, among other things, we rule out the existence of defiers and assume that $X$ is orthogonal to $Z$, the unconditional IV estimand, $\beta_{\mathrm{IV}} = \frac{\e \left[ Y \mid Z=1 \right] - \e \left[ Y \mid Z=0 \right]}{\e \left[ D \mid Z=1 \right] - \e \left[ D \mid Z=0 \right]}$, recovers the average treatment effect for compliers, also referred to as the local average treatment effect (LATE)\@.

Some of my results will allow for the existence of both compliers and defiers, and hence throughout this paper I instead follow \cite{Kolesar2013} in defining the LATE as
\begin{equation}
\label{eq:late1}
\tau_{\mathrm{LATE}} = \e \left[ Y(1) - Y(0) \mid D(1) \neq D(0) \right],
\end{equation}
\textit{i.e.}~the average treatment effect for individuals whose treatment status is affected by the instrument. This group includes both compliers and defiers; it will be restricted to compliers whenever the existence of defiers is ruled out. It is useful to note that this unconditional LATE parameter can also be written as
\begin{equation}
\label{eq:desired}
\tau_{\mathrm{LATE}} = \frac{\e \left[ \pi(X) \cdot \tau(X) \right]}{\e \left[ \pi(X) \right]},
\end{equation}
where
\begin{equation}
\tau(x) = \e \left[ Y(1) - Y(0) \mid D(1) \neq D(0), X=x \right]
\end{equation}
is the conditional LATE and
\begin{equation}
\pi(x) = \pr \left[ D(1) \neq D(0) \mid X=x \right]
\end{equation}
is the conditional proportion of compliers and defiers. The following assumption, together with additional assumptions below, will be used to identify $\tau(x)$ and $\pi(x)$, and thereby also $\tau_{\mathrm{LATE}}$\@. Recall that $Y = Y(D)$ when direct effects of $Z$ on $Y$ are ruled out and $Y = Y(Z,D)$ otherwise.

\begin{uassumption}{IV}
\label{ass:iv}
\begin{enumerate}
\item[\phantom{}]
\item[\textbf{(i)}] (Conditional independence) \; $\big( Y(0,0), Y(0,1), Y(1,0), Y(1,1), D(0), D(1) \big) \perp Z \mid X$;
\item[\textbf{(ii)}] (Exclusion restriction) \; $\pr \left[ Y(1,d)=Y(0,d) \mid X \right] = 1$ for $d \in \{ 0,1 \}$ a.s.;
\item[\textbf{(iii)}] (Relevance) \; $0 < \pr \left[ Z=1 \mid X \right] < 1$ and $\pr \left[ D(1)=1 \mid X \right] \neq \pr \left[ D(0)=1 \mid X \right]$ a.s.
\end{enumerate}
\end{uassumption}

\noindent
Assumption \ref{ass:iv} is standard but not sufficient to identify $\tau(x)$ and $\pi(x)$. It is also necessary to restrict the existence of defiers \citep{IA1994}. The following assumption, due to \cite{Abadie2003}, rules out the existence of defiers at any value of covariates.

\begin{uassumption}{SM}[Strong monotonicity]
\label{ass:um}
$\pr \left[ D(1) \geq D(0) \mid X \right] = 1$ a.s.
\end{uassumption}

\noindent
In many applications, Assumption \ref{ass:um} may be too restrictive \citep[cf.][]{deChaisemartin2017,DHM2023}. A testable implication of Assumption \ref{ass:um} is that $\omega(x)$, the conditional first-stage slope coefficient, is always non-negative. If this is formally rejected or otherwise implausible, an alternative assumption is necessary to obtain point identification. One possibility is to restrict treatment effect heterogeneity, as discussed by \cite{HV2005} and \cite{MT2018}, in which case we will be able to identify the average treatment effect rather than the unconditional LATE parameter. Another possibility is to replace Assumption \ref{ass:um} with a weaker assumption that postulates the existence of compliers but no defiers at some covariate values and the existence of defiers but no compliers elsewhere. While the relative appeal of these two assumptions is context dependent, I will focus on the latter in what follows.

\begin{uassumption}{WM}[Weak monotonicity]
\label{ass:cm}
There exists a subset of the support of $X$ such that $\pr \left[ D(1) \geq D(0) \mid X \right] = 1$ on it and $\pr \left[ D(1) \leq D(0) \mid X \right] = 1$ on its complement.
\end{uassumption}

\noindent
To understand the difference between Assumptions \ref{ass:um} and \ref{ass:cm}, consider a recent paper by \cite{DHMMR2019}, who estimate the health effects of air pollution using an instrument based on changes in local wind direction. Imagine a pollution source located to the east of a particular city. When the wind also blows from the east, the city will experience relatively high levels of pollution; the opposite is true when the wind blows from the west. Assumption \ref{ass:um} would require that every city reacts to a specific wind direction (say, east) in the same way (say, high pollution). This, however, is known not to be true. \cite{DHMMR2019} explain, for example, that air pollution is relatively high in San Francisco when the wind blows from the southeast, while the same is true in Boston when the wind blows from the southwest. Indeed, Assumption \ref{ass:cm} would allow for the possibility that different locations react to a specific wind direction in different ways.\footnote{If we knew the pollution-inducing wind direction for every location, as we do in the case of Boston and San Francisco, Assumption \ref{ass:um} might remain plausible for an appropriately redefined instrument. However, if this direction needs to be estimated, as is likely the case in practice, Assumption \ref{ass:cm} will be more appropriate.}

Importantly, Assumption \ref{ass:cm}, together with Assumption \ref{ass:iv}, is sufficient to identify $\tau(x)$ and $\pi(x)$. Before stating the relevant lemma, it is useful to define an auxiliary function
\begin{equation}
c(x) = \sgn \Big( \pr \left[ D(1) \geq D(0) \mid X=x \right] - \pr \left[ D(1) \leq D(0) \mid X=x \right] \Big),
\end{equation}
where $\sgn(\cdot)$ is the sign function. Clearly, $c(x)$ equals 1 if there are only compliers at $X=x$ and $-1$ if there are only defiers at $X=x$.

The following lemma summarizes identification of the conditional LATE parameter and the conditional proportion of individuals whose treatment status is affected by the instrument.

\begin{lemma}
\label{lem:identify}
\begin{enumerate}
\item[\phantom{}]
\item[\textbf{(i)}] Under Assumptions \ref{ass:iv} and \ref{ass:um}, $\tau(x) = \beta(x)$ and $\pi(x) = \omega(x)$.
\item[\textbf{(ii)}] Under Assumptions \ref{ass:iv} and \ref{ass:cm}, $\tau(x) = \beta(x)$ and $\pi(x) = \left\lvert \omega(x) \right\rvert = c(x) \cdot \omega(x)$.
\end{enumerate}
\end{lemma}

\noindent
Lemma \ref{lem:identify} consists of well-known results and straightforward extensions of these results, and as such it is stated without proof \citep[cf.][]{AIR1996,AP2009}. Note that strong monotonicity implies weak monotonicity, which means that every statement that is true under weak monotonicity is also true under strong monotonicity as a special case. I will follow this logic in the statement of the theoretical results below.

\section{Negative Weights in Linear IV}
\label{sec:theory}

\subsection{\cite{AI1995}, Revisited}
\label{sec:result1}

Let us begin by revisiting \citetalias{AI1995}'s representation of the 2SLS estimand. Recall that \citetalias{AI1995} study a special case of the model in equation (\ref{eq:structural}) where all covariates are binary and represent membership in disjoint groups or strata. In this case, each of the original covariates needs to be discrete or discretized, which means that the population can be divided into $K$ groups, where $K$ denotes the number of possible combinations of values of these variables. (For example, with six binary variables, we have $K = 2^6 = 64$.) Let $G \in \{ 1, \ldots, K \}$ denote group membership and $G_k = 1[G = k]$ denote the resulting group indicators. \citetalias{AI1995} consider a model where original covariates are replaced with these group indicators, $X = \left( 1, G_1, \ldots, G_{K-1} \right)$, while reduced-form and first-stage regressions include a full set of interactions between $X$ and $Z$; that is, $Z_{\mathrm{C}} = \left( Z, Z G_1, \ldots, Z G_{K-1} \right)$. The following lemma restates \citetalias{AI1995}'s and \cite{Kolesar2013}'s interpretation of the 2SLS estimand in this context.

\begin{lemma}[\citealp{AI1995,Kolesar2013}]
\label{lem:ai1995}
Suppose that $X = \left( 1, G_1, \ldots, G_{K-1} \right)$ and $Z_{\mathrm{C}} = \left( Z, Z G_1, \ldots, Z G_{K-1} \right)$. Suppose further that Assumptions \ref{ass:iv} and \ref{ass:cm} hold. Then
\begin{equation*}
\beta_{\mathrm{2SLS}} = \frac{\e \left[ \sigma^2(X) \cdot \tau(X) \right]}{\e \left[ \sigma^2(X) \right]},
\end{equation*}
where $\sigma^2(X) = \var \left[ \e \left[ D \mid X, Z \right] \mid X \right] = \e \left[ \left( \e \left[ D \mid X, Z \right] - \e \left[ D \mid X \right] \right) ^2 \mid X \right]$.
\end{lemma}

\noindent
Lemma \ref{lem:ai1995} establishes that the 2SLS estimand in \citetalias{AI1995}'s interacted specification is a convex combination of conditional LATEs, with weights equal to the conditional variance of the first stage. This result is due to \citetalias{AI1995} and has usually been interpreted as requiring that the existence of defiers is completely ruled out \citep[\textit{e.g.},][]{AP2009}. \cite{Kolesar2013} demonstrates that it also holds under weak monotonicity.

It may not be immediately obvious how the 2SLS weights in Lemma \ref{lem:ai1995} differ from the ``desired'' weights in equation (\ref{eq:desired}). The following result facilitates this comparison.

\begin{theorem}
\label{the:ai1995}
Suppose that $X = \left( 1, G_1, \ldots, G_{K-1} \right)$ and $Z_{\mathrm{C}} = \left( Z, Z G_1, \ldots, Z G_{K-1} \right)$. Suppose further that Assumptions \ref{ass:iv} and \ref{ass:cm} hold. Then
\begin{equation*}
\beta_{\mathrm{2SLS}} = \frac{\e \left[ \left[ \pi(X) \right] ^2 \cdot \var \left[ Z \mid X \right] \cdot \tau(X) \right]}{\e \left[ \left[ \pi(X) \right] ^2 \cdot \var \left[ Z \mid X \right] \right]}.
\end{equation*}
\end{theorem}

\noindent
Theorem \ref{the:ai1995} shows that the 2SLS estimand in \citetalias{AI1995}'s interacted specification is a convex combination of conditional LATEs, with weights equal to the product of the squared conditional proportion of compliers or defiers and the conditional variance of $Z$\@.\footnote{See also \cite{Walters2018} for a related remark that focuses on ``descriptive'' estimands and does not use the LATE framework for interpretation.} Since the ``desired'' weights in equation (\ref{eq:desired}) consist only of the conditional proportion of compliers or defiers,  \citetalias{AI1995}'s specification overweights the effects in groups with strong first stages and with large variances of $Z$\@. Importantly, this result does not require strong monotonicity; weak monotonicity is sufficient.

\begin{remark}
\label{rem:lackofoverid}
Although Lemma \ref{lem:ai1995} and Theorem \ref{the:ai1995} show that \citetalias{AI1995}'s specification can avoid negative weights, practitioners rarely use multiple interacted instruments. In a survey of recent applications of IV methods, \cite{BBMT2022,BBMT2025} determine that only 1 out of 99 applicable papers has used \citetalias{AI1995}'s specification. Specifications with many interactions between the instrument(s) and covariates were more common in earlier work using IV methods \citep[\textit{e.g.},][]{Angrist1990,AK1991} but have since become rare, likely out of concern for the many instrument bias.\footnote{Indeed, \cite{BJB1995} write that their results ``indicate that \textit{the common practice} of adding interaction terms as excluded instruments may exacerbate the problem'' (emphasis mine). On the other hand, some recent applications of the wind instrument \citep{DHMMR2019,BRS2020} and the ``judges design'' \citep{AD2015,MS2015,Stevenson2018} interact the instrument with selected covariates, which is similar in spirit to \citetalias{AI1995}'s specification. However, quantitatively speaking, this is still very rare in practice: in a sample of 122 papers considered by \cite{MTW2021}, only seven include specifications with some covariate interactions with a baseline instrument.}
\end{remark}

\subsection{Usual Application of IV}
\label{sec:result2}

Remark \ref{rem:lackofoverid} suggests that Theorem \ref{the:ai1995} cannot be used directly to interpret most empirical studies because modern applications of IV methods avoid using many interacted instruments. A similar point is made by \citet[][p.~178]{AP2009}, who maintain, however, that an indirect argument in \cite{Abadie2003} implies that ``some kind of covariate-averaged LATE'' is estimated in noninteracted specifications as well. In what follows, I show that \cite{AP2009}'s assertion would be \textit{false} under weak monotonicity. The claim is true under strong monotonicity, which I will be able to demonstrate directly, deriving the exact form of ``covariate-averaged LATE'' that linear IV estimates. I also revisit \cite{Abadie2003}'s indirect argument later on.

To save space, I combine two extensions of \citetalias{AI1995}'s analysis in what follows. On the one hand, I am interested in the interpretation of the IV estimand when we retain \citetalias{AI1995}'s restriction that the model for covariates is saturated but no longer use the interacted instruments. This analysis does not require any additional assumptions. On the other hand, I am also interested in the interpretation of the IV estimand in nonsaturated specifications. This analysis proceeds under the assumption that the instrument propensity score, defined as
\begin{equation}
e(X) = \e \left[ Z \mid X \right],
\end{equation}
is linear in $X$\@. This assumption is standard and has been used by \cite{Kolesar2013}, \cite{LM2015}, \cite{EK2019}, and \cite{Ishimaru2024}, among others.

\begin{uassumption}{PS}[Instrument propensity score]
\label{ass:ps}
$e(X) = X \alpha$.
\end{uassumption}

\noindent
Assumption \ref{ass:ps} holds automatically when $Z$ is randomized, and also when all covariates are discrete and the model for covariates is saturated. (This is why the statement of the theoretical results below only invokes Assumption \ref{ass:ps} and does not separately mention saturated specifications.) Assumption \ref{ass:ps} may also provide a good approximation to $e(X)$ in other situations, especially when $X$ includes powers and cross-products of the original covariates. This assumption is critical. \cite{BBMT2022,BBMT2025} determine that Assumption \ref{ass:ps} is necessary for the IV and 2SLS estimands to maintain their interpretation as a convex combination of conditional LATEs.

Let us first consider the case of weak monotonicity. The following result shows that the interpretation of the linear IV estimand is very unappealing in this context.

\begin{theorem}
\label{the:justid_cm}
Suppose that Assumptions \ref{ass:iv}, \ref{ass:cm}, and \ref{ass:ps} hold. Then
\begin{equation*}
\beta_{\mathrm{IV}} = \frac{\e \left[ c(X) \cdot \pi(X) \cdot \var \left[ Z \mid X \right] \cdot \tau(X) \right]}{\e \left[ c(X) \cdot \pi(X) \cdot \var \left[ Z \mid X \right] \right]}.
\end{equation*}
\end{theorem}

\noindent
Theorem \ref{the:justid_cm} provides a new representation of the IV estimand in the standard specification, \textit{i.e.}~one that, perhaps incorrectly, restricts the effects of the instrument in the reduced-form and first-stage regressions to be homogeneous across covariate values. Unlike in \citetalias{AI1995}'s specification, the estimand in the standard specification is not necessarily a convex combination of conditional LATEs. This is because $c(x)$ takes the value $-1$ for every value of covariates where there exist defiers but no compliers, and hence the corresponding weights in Theorem \ref{the:justid_cm} are negative as well. It follows that, when IV is applied in the usual way, the estimand may no longer be interpretable as a causal effect. It is even possible that this parameter may be negative (positive) when treatment effects are positive (negative) for everyone in the population.

The following result demonstrates that this problem disappears when we impose the strong version of monotonicity.

\begin{corollary}
\label{cor:justid_um}
Suppose that Assumptions \ref{ass:iv}, \ref{ass:um}, and \ref{ass:ps} hold. Then
\begin{equation*}
\beta_{\mathrm{IV}} = \frac{\e \left[ \pi(X) \cdot \var \left[ Z \mid X \right] \cdot \tau(X) \right]}{\e \left[ \pi(X) \cdot \var \left[ Z \mid X \right] \right]}.
\end{equation*}
\end{corollary}

\noindent
Corollary \ref{cor:justid_um} provides a direct argument for \cite{AP2009}'s assertion that the standard specification of IV recovers a convex combination of conditional LATEs. As noted previously, however, this statement is no longer true under weak monotonicity. If strong monotonicity holds, then the weights in Corollary \ref{cor:justid_um} may be more desirable than those in \citetalias{AI1995}'s specification. Indeed, a comparison of Corollary \ref{cor:justid_um} and equation (\ref{eq:desired}) shows that the standard specification, like \citetalias{AI1995}'s specification, overweights the effects in groups with large variances of $Z$ but not, unlike the latter, in groups with strong first stages.\footnote{To be clear, both specifications attach a greater weight to conditional LATEs in groups with strong first stages, as required by equation (\ref{eq:desired}). But \citetalias{AI1995}'s specification places even more weight on such conditional LATEs than is necessary to recover the unconditional LATE parameter.}

\begin{remark}
\label{rem:abadie}
\cite{Abadie2003} shows that, under Assumptions \ref{ass:iv}, \ref{ass:um}, and \ref{ass:ps}, the IV estimand is equivalent to the coefficient on $D$ in the linear projection of $Y$ on $D$ and $X$ among compliers. In other words, IV is analogous to ordinary least squares (OLS), with the exception of its ability to implicitly condition the analysis on the (latent) subpopulation of compliers. Corollary \ref{cor:justid_um} provides another argument that ``IV is like OLS\@.'' Indeed, as shown by \cite{Angrist1998}, the only difference between the OLS estimand and the ATE is in the dependence of the OLS weights on $\var \left[ D \mid X \right]$\@. Similarly, Corollary \ref{cor:justid_um} shows that, under strong monotonicity, the only difference between the IV estimand and the LATE is in the dependence of the IV weights on $\var \left[ Z \mid X \right]$\@. However, this analogy between OLS and IV may be problematic for IV given the undesirable properties of the OLS estimand under treatment effect heterogeneity \citep[cf.][]{Sloczynski2022}.
\end{remark}

\begin{remark}
\label{rem:bww2007}
\cite{BWW2007} discuss the interpretation of interacted and noninteracted specifications in randomized experiments with noncompliance in which the existence of defiers is completely ruled out. In this case, the standard specification of IV recovers the unconditional LATE parameter but the interacted specification does not.\footnote{Instead, the interacted specification recovers a convex combination of conditional LATEs, which is generally different from the unconditional LATE parameter. A similar point about models with fully independent instruments is made by \cite{HK2020}, who also revisits the link between the existence of defiers and negative weights in this context \citep[cf.][]{IA1994,deChaisemartin2017,DHM2023} and recommends interacted specifications.} This is a special case of the difference between Theorem \ref{the:ai1995} and Corollary \ref{cor:justid_um} where $\var \left[ Z \mid X \right]$ is constant. However, Theorem \ref{the:justid_cm} makes it clear that under weak monotonicity the standard specification no longer recovers the unconditional LATE parameter or even a convex combination of conditional LATEs.
\end{remark}

\begin{remark}
\label{rem:kolesar}
Theorem \ref{the:justid_cm} and Corollary \ref{cor:justid_um} are also related to Theorem 1 in \cite{Kolesar2013}, which provides a common representation of any two-step instrumental variables estimand in the case of a binary $D$, a discrete $Z$, and under conditions similar to Assumptions \ref{ass:iv}, \ref{ass:cm}, and \ref{ass:ps}\@. To present this result, it is necessary to introduce some additional notation. Let $P = \e \left[ D \mid Z,X \right]$, $P^L = \mathrm{L} \left[ D \mid Z_{\mathrm{G}},X \right]$, and $\tilde{P}^L = P^L - \mathrm{L} \left[ D \mid X \right]$, where $\mathrm{L}[\cdot]$ is the linear projection and $Z_{\mathrm{G}} = z_{\mathrm{G}}(X,Z)$ is the vector of constructed instruments, which may include (some) interactions between $X$ and $Z$\@. Also, let $\mathcal{P}_x$ denote the support of $P$ conditional on $X=x$ and $J_x$ denote the number of support points, with $\mathcal{P}_x = \big\{ p_{1,x} < \ldots < p_{J_x,x} \big\}$. Then, \cite{Kolesar2013} shows that
\begin{equation}
\label{eq:kolesar}
\beta_{\mathrm{TSIV}} = \int \sum_{j=1}^{J_x-1} \frac{\theta_j(x)}{\int \sum_{j=1}^{J_x-1} \theta_j(x) \, \mathrm{d} F^X(x)} \, \tau(p_{j,x};x) \, \mathrm{d} F^X(x),
\end{equation}
where $\beta_{\mathrm{TSIV}}$ is any two-step instrumental variables estimand (\textit{e.g.}, 2SLS) which uses $Z_{\mathrm{G}}$ as instruments, $\theta_j(x) = \left( p_{j+1,x} - p_{j,x} \right) \cdot \pr \left[ P > p_{j,x} \mid X=x \right] \cdot \e \left[ \tilde{P}^L \mid X=x, P > p_{j,x} \right]$, and $\tau(p_{j,x};x) = \frac{\e \left[ Y \, \mid \, P=p_{j+1,x}, \, X=x \right] \; - \; \e \left[ Y \, \mid \, P=p_{j,x}, \, X=x \right]}{p_{j+1,x} \; - \; p_{j,x}}$ is the conditional LATE based on two adjacent elements of $\mathcal{P}_x$. \cite{Kolesar2013}'s result is generic in the sense that it applies to any given vector of instruments $Z_{\mathrm{G}} = z_{\mathrm{G}}(X,Z)$\@. At the same time, Theorem \ref{the:justid_cm} is specific to the IV estimand. However, its focus on that particular specification simplifies the result, making it more transparent and easier to interpret than equation (\ref{eq:kolesar}).\footnote{Using equation (\ref{eq:kolesar}) to determine whether a given specification rules out the incidence of negative weights requires verifying the condition $\pr \left[ \theta_j(X) \geq 0 \right] = 1$ on a case-by-case basis.} In Appendix \ref{app:proofs}, I also present an alternative proof of Theorem \ref{the:justid_cm}, which uses \cite{Kolesar2013}'s representation of $\beta_{\mathrm{TSIV}}$.
\end{remark}

\begin{remark}
\label{rem:montest}
A testable implication of strong monotonicity is that $\omega(x)$, the conditional first-stage slope coefficient, is always non-negative. In a saturated specification with $X = \left( 1, G_1, \ldots, G_{K-1} \right)$, it is straightforward to construct a formal test based on this observation.\footnote{See also \cite{Semenova2025} for an analogous test in the context of endogenous sample selection.} If we define
\begin{equation}
\omega = \Big( \e \left[ D \mid Z=1, G=k \right] - \e \left[ D \mid Z=0, G=k \right] \Big)_{k=1}^{K},
\end{equation}
then the null hypothesis can be written as
\begin{equation}
H_0 : \quad \left( -1 \right) \cdot \omega \leq 0
\end{equation}
and the test statistic as
\begin{equation}
\label{eq:supt}
T = \max_{1 \leq k \leq K} \frac{\left( -1 \right) \cdot \hat{\omega}_k}{\hat{\sigma}_{\hat{\omega}_k}}.
\end{equation}
One possible choice of critical values for this test statistic are the one-step self-normalized critical values of \cite{CCK2019}. Another is based on the Bonferroni procedure, which requires, however, that $K$ is much smaller than the sample size.
\end{remark}

\begin{remark}
\label{rem:riv}
Suppose we are interested in the estimand of Corollary \ref{cor:justid_um}, but we are only willing to assume weak monotonicity. If $\omega(x)$ were known, we could define a new, ``reordered'' instrument as $Z_{\mathrm{R}} = 1 [ \omega(X) > 0 ] \cdot Z + 1 [ \omega(X) < 0 ] \cdot \left( 1-Z \right)$ and subsequently use it in a noninteracted specification. In Appendix \ref{app:proofs}, I show that this procedure would recover the estimand of interest. In practice, however, $\omega(x)$ is unknown and would need to be estimated. I leave the study of the properties of the resulting reordered IV estimator to future work.
\end{remark}

\subsection{Finite Sample Considerations}
\label{sec:manyiv}

Given the theoretical results in Sections \ref{sec:result1} and \ref{sec:result2}, it seems reasonable to consider \citetalias{AI1995}'s interacted specification whenever weak monotonicity is plausible but strong monotonicity is not. However, this approach has some limitations in finite samples: it requires dividing the sample into $K$ groups, and when $K$ is sufficiently large relative to the sample size, some groups will be small. With many groups and instruments, this situation leads to bias, which results from overfitting the first stage. In other words, the first-stage fitted values pick up the noise, not just the signal, and a large amount of noise, particularly likely with many small groups, translates to poor estimates of the first stage and bias in the second stage.

This phenomenon, known as the ``many instrument'' bias, has been extensively studied in the econometrics literature. Recent surveys include \cite{Anatolyev2019} and \cite{MS2024}.\footnote{The classic literature on many instruments has focused on the homogeneous effects model, but I interpret its results through the lens of the framework in Section \ref{sec:framework}.} In the remainder of this section, I first review several solutions to this problem, which offer finite sample improvements over 2SLS when estimating specifications with many instruments (\textit{e.g.}, \citetalias{AI1995}'s specification). Then, I review a recent pretest designed to evaluate whether, in a given dataset, the instruments are jointly strong enough to ensure consistency. I conclude with a simulation study.

\subsubsection{Estimation with Many Instruments}

The problem of the many instrument bias is usually studied using the asymptotic sequence of \cite{Kunitomo1980}, \cite{Morimune1983}, and \cite{Bekker1994}, which allows the number of instruments, $K$, to increase in proportion with the sample size, $N$\@. In the context of \citetalias{AI1995}'s interacted specification, fixing the ratio of $K$ to $N$ does not allow the group sizes to grow when the sample size grows, which reproduces the practical problem of small groups.

Under this asymptotic sequence, 2SLS is inconsistent unless the concentration parameter, a measure of instrument strength, grows faster than the number of instruments. The classic alternatives include the limited information maximum likelihood (LIML) estimator of \cite{AR1949} and the bias-corrected two-stage least squares (B2SLS) estimator of \cite{Nagar1959}, both of which are consistent under homoskedasticity when the concentration parameter grows faster than the square root of the number of instruments \citep{CS2005}. However, homoskedasticity of first-stage errors is impossible when the treatment is binary. Under heteroskedasticity, LIML and B2SLS require the same (stronger) condition as 2SLS \citep{CSHNW2012}.

Under heteroskedasticity, the weaker condition that the concentration parameter grows faster than the square root of the number of instruments is sufficient for the consistency of the jackknife IV estimator (JIVE) of \cite{AIK1999}, as also shown by \cite{CSHNW2012}. The basic idea underlying jackknife-type estimators is that using a ``leave-one-out'' predictor of the treatment---effectively a separate first stage for each unit---will reduce the noise and bias.

At the same time, however, most of the estimators discussed so far are inconsistent under the asymptotic sequence that allows the number of covariates, alongside the number of instruments, to increase in proportion with the sample size. This is potentially a major limitation because, in \citetalias{AI1995}'s specification, the number of covariates and the number of instruments are the same and equal to the number of groups. Still, several modifications to JIVE and B2SLS are robust to many instruments and many covariates, including the improved jackknife IV estimator (IJIVE) of \cite{AD2009}, the modified bias-corrected two-stage least squares (MB2SLS) estimator of \cite{Anatolyev2013}, the unbiased jackknife IV estimator (UJIVE) of \cite{Kolesar2013}, and three jackknife-type estimators of \cite{CSW2023}, referred to as the fixed effect jackknife IV (FEJIV) estimator, the fixed effect limited information maximum likelihood (FELIM) estimator, and the fixed effect \cite{Fuller1977} (FEFUL) estimator. Although the performance of LIML is not additionally affected by many covariates \citep{Anatolyev2013}, both LIML and MB2SLS rely on the homoskedasticity assumption. Furthermore, LIML does not even share the estimand with two-step IV estimators, such as 2SLS, MB2SLS, JIVE, IJIVE, UJIVE, and FEJIV, making it inappropriate in settings with treatment effect heterogeneity \citep{Kolesar2013}. FELIM and FEFUL do not belong to the class of two-step IV estimators either. Finally, \cite{CSW2023} discuss the limitations of IJIVE and, to a lesser extent, UJIVE, making FEJIV the likely estimator of choice.

While the framework of \cite{CSW2023} does not explicitly allow for treatment effect heterogeneity, the suitability of the FEJIV estimator in my framework follows from \cite{Kolesar2013}, who shows that any member of a broad class of two-step IV estimators has a common weighted average representation under treatment effect heterogeneity (cf.~Remark \ref{rem:kolesar}). Because both 2SLS and FEJIV fall into this class, their estimands have the same interpretation under treatment effect heterogeneity and standard asymptotics. The difference is that under many instrument asymptotics, 2SLS becomes inconsistent for this estimand, whereas FEJIV remains consistent.

\subsubsection{Weak Identification}

Specifications with many instruments require that they are sufficiently strong as a group, although they can be individually weak or even irrelevant \citep[cf.][]{Anatolyev2019}. In the context of \citetalias{AI1995}'s specification, the original instrument can be weak in some groups as long as it is sufficiently strong in others. But how strong is strong enough?

\cite{MS2022} study weak identification in linear models with many instruments, which is a situation where the concentration parameter divided by the square root of the number of instruments remains bounded as the sample size grows. They also develop a pretest for this phenomenon to evaluate whether identification is strong in a given dataset. (Their test statistic $\widetilde{F}$ should be compared to a cutoff of 4.14.) Under the null of weak identification, no consistent estimator exists, and inference can instead be based on a jackknifed version of the AR test statistic. When the pretest rejects, \cite{MS2022} recommend the jackknife IV estimator, which is consistent under the alternative \citep{CSHNW2012}.

\subsubsection{Simulations}
\label{sec:simul}

In what follows, I study the finite sample performance of several two-step IV estimators of \citetalias{AI1995}'s specification, with a focus on settings with many small groups, treatment effect heterogeneity, and violations of Assumption \ref{ass:um}\@. I adapt the data-generating process from \cite{BBMT2022}, which was designed to mimic the college proximity study in \cite{Card1995}. The simulation design also originally assumed homogeneous treatment effects and no monotonicity violations. As we will see, these restrictions are responsible for \cite{BBMT2022}'s conclusion that the usual application of IV is easier to estimate without bias than the interacted specification.

In the baseline data-generating process, as in \cite{BBMT2022}, I draw $X$ uniformly from a Halton sequence $\mathcal{X}$ on $[0,1]$, subsequently drawing $Z$, $D(Z)$, and $Y(D)$ as
\begingroup
\allowdisplaybreaks
\begin{eqnarray}
\pr \left[ Z=1 \mid X \right] &=& 0.119 + 1.785 X - 1.534 X^2 + 0.597 X^3, \\
D(Z) &=& 1 [ \Phi(V) \leq p(Z) ], \\
Y(D) &=& \log \left( 129.7 + 1247.7 X - 2149 X^2 + 1515.7 X^3 \right) + 1.2 D + U,
\end{eqnarray}
\endgroup
where $(U,V)$ are standard multivariate normal with correlation 0.527, drawn independently of $(X,Z)$\@. I also set $|\mathcal{X}| = 250$, $p(0) = \pr \left[ D=1 \mid Z=0 \right] = 0.22$, and $p(1) = \pr \left[ D=1 \mid Z=1 \right] = 0.29$. In this setting, treatment effects are homogeneous and equal to 1.2. Strong monotonicity is satisfied even though the instrument is relatively weak, with the proportion of compliers independent of $X$ and equal to $p(1) - p(0) = 0.07$. Again, these parameters are calibrated to the data in \cite{Card1995}.

In subsequent modifications of this data-generating process, I introduce treatment effect heterogeneity by specifying $Y(1)$ and $Y(0)$ as
\begingroup
\allowdisplaybreaks
\begin{eqnarray}
Y(1) &=& \log \left( 129.7 + 1247.7 X - 2149 X^2 + 1515.7 X^3 \right) + 1.2 + U, \\
Y(0) &=& \log \left( 1 \cdot 129.7 + 2 \cdot 1247.7 X - 3 \cdot 2149 X^2 + 4 \cdot 1515.7 X^3 \right) + U,
\end{eqnarray}
\endgroup
while also allowing for violations of strong (but not weak) monotonicity. This is accomplished by switching the values of $p(0)$ and $p(1)$ for some groups. Specifically, to generate what I refer to as ``moderate'' monotonicity violations, I reverse the values of $p(0)$ and $p(1)$ if $X>0.75$. For ``large'' monotonicity violations, the threshold value of $X$ is 0.5. I also consider a setting with ``weak cells,'' that is, values of $X$ where the proportion of compliers and defiers is zero. Here, I set $p(0) = p(1) = 0.22$ if $1/3 < X < 2/3$ and reverse the original values of $p(0)$ and $p(1)$ if $X>2/3$.

Two final modifications involve the number and relative sizes of groups and the instrument strength. So far, the groups were equal sized. To reproduce the likely scenario that some groups are large while others are small, I also consider a setting with $|\mathcal{X}| = 20$, but where $X$ is not drawn uniformly. Specifically, I set $\pr \left[ G=k \right]$ to be proportional to $1.3^k$, making the largest group $1.3^{19}$ times larger than the smallest. As in \cite{BBMT2022}, I also consider a scenario where the instrument is stronger than the ``weak'' case above, with 0.52 replacing 0.29 as the larger value of $p(Z)$ whenever $p(0) \neq p(1)$ conditional on $X$\@. This sets the conditional proportion of compliers or defiers equal to 0.3, except in the ``weak cells'' design, where it is either 0.3 or 0.

\begin{table}[!t]
\begin{adjustwidth}{-1in}{-1in}
\centering
\begin{threeparttable}
\caption{Simulation Results for $K=250$, ``Weak'' IV, and No Monotonicity Violations\label{tab:sim_250_1}}
\begin{footnotesize}
\begin{tabular}{c>{\centering\arraybackslash}m{0.05cm} >{\centering\arraybackslash}m{1.1425cm} >{\centering\arraybackslash}m{1.1425cm} >{\centering\arraybackslash}m{1.1425cm} >{\centering\arraybackslash}m{0.05cm} >{\centering\arraybackslash}m{1.1425cm} >{\centering\arraybackslash}m{1.1425cm} >{\centering\arraybackslash}m{1.1425cm} >{\centering\arraybackslash}m{0.05cm} >{\centering\arraybackslash}m{1.1425cm} >{\centering\arraybackslash}m{1.1425cm} >{\centering\arraybackslash}m{1.1425cm}}
\hline\hline
          &       & \multicolumn{3}{c}{$N=3{,}000$} &       & \multicolumn{3}{c}{$N=10{,}000$} &       & \multicolumn{3}{c}{$N=50{,}000$} \\
\hline
    A. Estimator Performance &       & Bias  & Median Bias & MSE   &       & Bias  & Median Bias & MSE   &       & Bias  & Median Bias & MSE \\
\hline
    OLS   &       & --0.735 & --0.735 & 1.080  &       & --0.736 & --0.735 & 1.303 &       & --0.735 & --0.735 & 3.457 \\
    IV    &       & 0.095 & 0.001 & 0.941 &       & 0.013 & --0.002 & 0.225 &       & 0.001 & --0.001 & 0.106 \\
    2SLS  &       & --0.698 & --0.701 & 1.000     &       & --0.635 & --0.636 & 1.000     &       & --0.387 & --0.387 & 1.000 \\
    MB2SLS &       & --59.10 & --0.966 & 6.4e+06 &       & 4.546 & --0.025 & 7.1e+04 &       & 0.053 & 0.036 & 0.313 \\
    JIVE  &       & --0.803 & --0.805 & 1.355 &       & --0.901 & --0.892 & 2.029 &       & --9.473 & --3.842 & 1.5e+05 \\
    IJIVE &       & 0.385 & --0.430 & 629.8 &       & --0.068 & --0.112 & 37.97 &       & 0.013 & --0.004 & 0.256 \\
    UJIVE &       & --1.540 & --0.482 & 516.7 &       & 1.073 & --0.062 & 4.8e+03  &       & 0.017 & 0.000     & 0.261 \\
    FEJIV &       & --1.396 & --0.512 & 3.5e+03  &       & 0.395 & --0.027 & 132.0   &       & 0.017 & 0.000     & 0.261 \\
          &       &       &       &       &       &       &       &       &       &       &       &  \\
\hline
    B. Pretest for Weak Identification &       &       &       &       &       &       &       &       &       &       &       &  \\
\hline
    Average $\widetilde{F}$ &       &       & 1.83  &       &       &       & 2.27  &       &       &       & 10.37 &  \\
    $q_{0.05}$ &       &       & --0.07 &       &       &       & 0.12  &       &       &       & 7.54  &  \\
    $q_{0.95}$ &       &       & 3.83  &       &       &       & 4.40   &       &       &       & 13.43 &  \\
\hline
\end{tabular}
\begin{tablenotes}[flushleft]
\item \textit{Notes:} ``OLS'' is the OLS estimator in the regression of the outcome on the treatment indicator and group indicators. ``IV'' is the IV estimator in the noninteracted specification. The remaining estimators are based on the interacted specification. JIVE, IJIVE, and UJIVE are computed after dropping all groups with fewer than two observations in either $(X,Z)$ combination. FEJIV is computed after dropping all groups with fewer than three observations in either $(X,Z)$ combination. The pretest for weak identification follows \cite{MS2022}; see also the Stata implementation in \cite{Sun2023}. Bias and median bias are reported as the proportion of the target parameter. MSE is normalized by the MSE of 2SLS\@. Results are based on 1,000 replications. Pretest results are based on 250 replications.
\end{tablenotes}
\end{footnotesize}
\end{threeparttable}
\end{adjustwidth}
\end{table}

The total number of simulation designs is sixteen, with $|\mathcal{X}| = 20$ or $|\mathcal{X}| = 250$, two levels of instrument strength (``weak'' or ``strong''), and four scenarios of violations of strong monotonicity, referred to as no violations, moderate violations, large violations, and violations with weak cells. Treatment effects are homogeneous when strong monotonicity holds and heterogeneous otherwise. The target parameter is the estimand in Theorem \ref{the:ai1995}, which is, except in the ``weak cells'' design, equal to that in Corollary \ref{cor:justid_um}, making monotonicity violations the only reason why the estimands of the interacted and noninteracted specifications may be different. I consider two sample sizes, $N=3{,}000$ and $N=10{,}000$, when $|\mathcal{X}| = 20$, and additionally $N=50{,}000$ when $|\mathcal{X}| = 250$. The smallest sample size, $N=3{,}000$, is similar to the sample size in \cite{Card1995}.

\begin{table}[!t]
\begin{adjustwidth}{-1in}{-1in}
\centering
\begin{threeparttable}
\caption{Simulation Results for $K=250$, ``Weak'' IV, and Moderate Monotonicity Violations\label{tab:sim_250_2}}
\begin{footnotesize}
\begin{tabular}{c>{\centering\arraybackslash}m{0.05cm} >{\centering\arraybackslash}m{1.1425cm} >{\centering\arraybackslash}m{1.1425cm} >{\centering\arraybackslash}m{1.1425cm} >{\centering\arraybackslash}m{0.05cm} >{\centering\arraybackslash}m{1.1425cm} >{\centering\arraybackslash}m{1.1425cm} >{\centering\arraybackslash}m{1.1425cm} >{\centering\arraybackslash}m{0.05cm} >{\centering\arraybackslash}m{1.1425cm} >{\centering\arraybackslash}m{1.1425cm} >{\centering\arraybackslash}m{1.1425cm}}
\hline\hline
          &       & \multicolumn{3}{c}{$N=3{,}000$} &       & \multicolumn{3}{c}{$N=10{,}000$} &       & \multicolumn{3}{c}{$N=50{,}000$} \\
\hline
    A. Estimator Performance &       & Bias  & Median Bias & MSE   &       & Bias  & Median Bias & MSE   &       & Bias  & Median Bias & MSE \\
\hline
    OLS   &       & --1.159 & --1.158 & 1.168 &       & --1.160 & --1.159 & 1.316 &       & --1.160 & --1.159 & 3.357 \\
    IV    &       & 0.543 & 0.181 & 3.418 &       & 0.219 & 0.176 & 0.440  &       & 0.180  & 0.175 & 0.239 \\
    2SLS  &       & --1.059 & --1.056 & 1.000     &       & --0.999 & --0.997 & 1.000     &       & --0.621 & --0.623 & 1.000 \\
    MB2SLS &       & --2.382 & --1.589 & 452.1 &       & 0.174 & --0.113 & 140.2 &       & --0.019 & --0.037 & 0.213 \\
    JIVE  &       & --1.150 & --1.152 & 1.215 &       & --1.403 & --1.391 & 1.997 &       & --9.380 & --5.865 & 6.9e+04 \\
    IJIVE &       & --2.173 & --0.672 & 1.7e+03  &       & 4.456 & --0.161 & 1.8e+04 &       & 0.020  & 0.001 & 0.229 \\
    UJIVE &       & --1.436 & --0.793 & 2.1e+03  &       & --0.735 & --0.076 & 812.2 &       & 0.028 & 0.007 & 0.234 \\
    FEJIV &       & --0.991 & --0.715 & 895.6 &       & 0.807 & --0.072 & 261.1 &       & 0.028 & 0.007 & 0.234 \\
          &       &       &       &       &       &       &       &       &       &       &       &  \\
\hline
    B. Pretest for Weak Identification &       &       &       &       &       &       &       &       &       &       &       &  \\
\hline
    Average $\widetilde{F}$ &       &       & 1.92  &       &       &       & 2.30   &       &       &       & 10.08 &  \\
    $q_{0.05}$ &       &       & --0.01 &       &       &       & 0.27  &       &       &       & 7.48  &  \\
    $q_{0.95}$ &       &       & 4.04  &       &       &       & 4.73  &       &       &       & 12.87 &  \\
\hline
\end{tabular}
\begin{tablenotes}[flushleft]
\item \textit{Notes:} ``OLS'' is the OLS estimator in the regression of the outcome on the treatment indicator and group indicators. ``IV'' is the IV estimator in the noninteracted specification. The remaining estimators are based on the interacted specification. JIVE, IJIVE, and UJIVE are computed after dropping all groups with fewer than two observations in either $(X,Z)$ combination. FEJIV is computed after dropping all groups with fewer than three observations in either $(X,Z)$ combination. The pretest for weak identification follows \cite{MS2022}; see also the Stata implementation in \cite{Sun2023}. Bias and median bias are reported as the proportion of the target parameter. MSE is normalized by the MSE of 2SLS\@. Results are based on 1,000 replications. Pretest results are based on 250 replications.
\end{tablenotes}
\end{footnotesize}
\end{threeparttable}
\end{adjustwidth}
\end{table}

Table \ref{tab:sim_250_1} reports simulation results for a number of estimators in the ``weak'' IV case with 250 groups and no monotonicity violations. The first three columns, setting $N=3{,}000$, correspond to the baseline results in \cite{BBMT2022}. Even though I consider a larger number of estimators than \cite{BBMT2022}, I reach the same conclusion: all estimators are severely biased, with the only exception of IV in the noninteracted specification, whose bias is less than 10\% and median bias is practically zero. However, panel B of Table \ref{tab:sim_250_1} reveals that this conclusion is predictable: the average value of \cite{MS2022}'s test statistic, $\widetilde{F}$, is 1.83, well below the cutoff of 4.14, which means that consistent estimation of the interacted specification is impossible. The remaining columns report simulation results for $N=10{,}000$ and $N=50{,}000$. Here, the strength of identification gradually increases, with the average value of $\widetilde{F}$ exceeding 10 when $N=50{,}000$. Indeed, when this is the case, the best-performing estimators of \citetalias{AI1995}'s specification---IJIVE, UJIVE, and FEJIV---are practically unbiased, in line with the results in \cite{MS2022}.

Table \ref{tab:sim_250_2} introduces moderate monotonicity violations. With $N=3{,}000$, the average value of $\widetilde{F}$ is again below 2. Now, however, every estimator is severely biased, including IV in the noninteracted specification. (Estimation of this specification is biased because of monotonicity violations. Estimation of the interacted specification is biased because of insufficient instrument strength.) With larger sample sizes, $N=10{,}000$ and $N=50{,}000$, identification gets stronger. Specifically, when $N=50{,}000$, the average value of $\widetilde{F}$ again exceeds 10, and IJIVE, UJIVE, and FEJIV perform very well. IV estimation of the noninteracted specification remains biased; however, it is competitive with the best-performing estimators in terms of MSE\@.

\begin{table}[!p]
\begin{adjustwidth}{-1in}{-1in}
\centering
\begin{threeparttable}
\caption{Simulation Results for $K=250$, ``Weak'' IV, and Large Monotonicity Violations\label{tab:sim_250_3}}
\begin{footnotesize}
\begin{tabular}{c>{\centering\arraybackslash}m{0.05cm} >{\centering\arraybackslash}m{1.1425cm} >{\centering\arraybackslash}m{1.1425cm} >{\centering\arraybackslash}m{1.1425cm} >{\centering\arraybackslash}m{0.05cm} >{\centering\arraybackslash}m{1.1425cm} >{\centering\arraybackslash}m{1.1425cm} >{\centering\arraybackslash}m{1.1425cm} >{\centering\arraybackslash}m{0.05cm} >{\centering\arraybackslash}m{1.1425cm} >{\centering\arraybackslash}m{1.1425cm} >{\centering\arraybackslash}m{1.1425cm}}
\hline\hline
          &       & \multicolumn{3}{c}{$N=3{,}000$} &       & \multicolumn{3}{c}{$N=10{,}000$} &       & \multicolumn{3}{c}{$N=50{,}000$} \\
\hline
    A. Estimator Performance &       & Bias  & Median Bias & MSE   &       & Bias  & Median Bias & MSE   &       & Bias  & Median Bias & MSE \\
\hline
    OLS   &       & --1.165 & --1.166 & 1.187 &       & --1.166 & --1.166 & 1.329 &       & --1.166 & --1.166 & 3.347 \\
    IV    &       & 0.839 & --0.252 & 723.7 &       & 1.498 & 0.459 & 252.2 &       & 0.670  & 0.547 & 3.322 \\
    2SLS  &       & --1.056 & --1.057 & 1.000     &       & --1.001 & --1.006 & 1.000     &       & --0.627 & --0.624 & 1.000 \\
    MB2SLS &       & 0.073 & --0.751 & 625.9 &       & --0.052 & --0.291 & 15.78 &       & --0.070 & --0.082 & 0.186 \\
    JIVE  &       & --1.154 & --1.151 & 1.239 &       & --1.413 & --1.399 & 2.025 &       & --2.757 & --5.762 & 1.6e+04 \\
    IJIVE &       & --1.341 & --0.753 & 456.7 &       & 0.085 & --0.212 & 158.8 &       & 0.018 & 0.004 & 0.216 \\
    UJIVE &       & --8.712 & --0.842 & 7.0e+04 &       & 0.011 & --0.113 & 99.13 &       & 0.026 & 0.010  & 0.221 \\
    FEJIV &       & --1.745 & --0.686 & 6.3e+03  &       & 1.626 & --0.057 & 721.5 &       & 0.026 & 0.010  & 0.221 \\
          &       &       &       &       &       &       &       &       &       &       &       &  \\
\hline
    B. Pretest for Weak Identification &       &       &       &       &       &       &       &       &       &       &       &  \\
\hline
    Average $\widetilde{F}$ &       &       & 2.01  &       &       &       & 2.24  &       &       &       & 9.74  &  \\
    $q_{0.05}$ &       &       & --0.06 &       &       &       & 0.37  &       &       &       & 7.15  &  \\
    $q_{0.95}$ &       &       & 4.30   &       &       &       & 4.41  &       &       &       & 12.41 &  \\
\hline
\end{tabular}
\begin{tablenotes}[flushleft]
\item \textit{Notes:} ``OLS'' is the OLS estimator in the regression of the outcome on the treatment indicator and group indicators. ``IV'' is the IV estimator in the noninteracted specification. The remaining estimators are based on the interacted specification. JIVE, IJIVE, and UJIVE are computed after dropping all groups with fewer than two observations in either $(X,Z)$ combination. FEJIV is computed after dropping all groups with fewer than three observations in either $(X,Z)$ combination. The pretest for weak identification follows \cite{MS2022}; see also the Stata implementation in \cite{Sun2023}. Bias and median bias are reported as the proportion of the target parameter. MSE is normalized by the MSE of 2SLS\@. Results are based on 1,000 replications. Pretest results are based on 250 replications.
\end{tablenotes}
\end{footnotesize}
\end{threeparttable}
\end{adjustwidth}
\end{table}

\begin{table}[!p]
\begin{adjustwidth}{-1in}{-1in}
\centering
\begin{threeparttable}
\caption{Simulation Results for $K=250$, ``Weak'' IV, and Monotonicity Violations with Weak Cells\label{tab:sim_250_4}}
\begin{footnotesize}
\begin{tabular}{c>{\centering\arraybackslash}m{0.05cm} >{\centering\arraybackslash}m{1.1425cm} >{\centering\arraybackslash}m{1.1425cm} >{\centering\arraybackslash}m{1.1425cm} >{\centering\arraybackslash}m{0.05cm} >{\centering\arraybackslash}m{1.1425cm} >{\centering\arraybackslash}m{1.1425cm} >{\centering\arraybackslash}m{1.1425cm} >{\centering\arraybackslash}m{0.05cm} >{\centering\arraybackslash}m{1.1425cm} >{\centering\arraybackslash}m{1.1425cm} >{\centering\arraybackslash}m{1.1425cm}}
\hline\hline
          &       & \multicolumn{3}{c}{$N=3{,}000$} &       & \multicolumn{3}{c}{$N=10{,}000$} &       & \multicolumn{3}{c}{$N=50{,}000$} \\
\hline
    A. Estimator Performance &       & Bias  & Median Bias & MSE   &       & Bias  & Median Bias & MSE   &       & Bias  & Median Bias & MSE \\
\hline
    OLS   &       & --1.196 & --1.195 & 1.136 &       & --1.193 & --1.194 & 1.185 &       & --1.194 & --1.194 & 2.269 \\
    IV    &       & 0.945 & --0.310 & 975.8 &       & 0.053 & 0.380  & 1.2e+03  &       & 0.805 & 0.646 & 3.480 \\
    2SLS  &       & --1.108 & --1.111 & 1.000     &       & --1.082 & --1.080 & 1.000     &       & --0.779 & --0.784 & 1.000 \\
    MB2SLS &       & --6.127 & --1.278 & 1.2e+04 &       & 0.183 & --0.479 & 246.3 &       & --0.109 & --0.132 & 0.308 \\
    JIVE  &       & --1.129 & --1.122 & 1.084 &       & --1.315 & --1.319 & 1.493 &       & --3.229 & --2.828 & 23.50 \\
    IJIVE &       & --11.17 & --0.900 & 2.1e+05 &       & 9.651 & --0.480 & 5.2e+04 &       & 0.026 & --0.015 & 0.396 \\
    UJIVE &       & --0.814 & --1.063 & 3.3e+03  &       & 3.604 & --0.411 & 1.1e+04 &       & 0.039 & --0.001 & 0.409 \\
    FEJIV &       & 3.766 & --0.778 & 2.3e+04 &       & --6.777 & --0.476 & 1.6e+04 &       & 0.039 & 0.001 & 0.408 \\
          &       &       &       &       &       &       &       &       &       &       &       &  \\
\hline
    B. Pretest for Weak Identification &       &       &       &       &       &       &       &       &       &       &       &  \\
\hline
    Average $\widetilde{F}$ &       &       & 1.73  &       &       &       & 1.48  &       &       &       & 6.18  &  \\
    $q_{0.05}$ &       &       & 0.03  &       &       &       & --0.58 &       &       &       & 4.10   &  \\
    $q_{0.95}$ &       &       & 3.57  &       &       &       & 3.54  &       &       &       & 8.52  &  \\
\hline
\end{tabular}
\begin{tablenotes}[flushleft]
\item \textit{Notes:} ``OLS'' is the OLS estimator in the regression of the outcome on the treatment indicator and group indicators. ``IV'' is the IV estimator in the noninteracted specification. The remaining estimators are based on the interacted specification. JIVE, IJIVE, and UJIVE are computed after dropping all groups with fewer than two observations in either $(X,Z)$ combination. FEJIV is computed after dropping all groups with fewer than three observations in either $(X,Z)$ combination. The pretest for weak identification follows \cite{MS2022}; see also the Stata implementation in \cite{Sun2023}. Bias and median bias are reported as the proportion of the target parameter. MSE is normalized by the MSE of 2SLS\@. Results are based on 1,000 replications. Pretest results are based on 250 replications.
\end{tablenotes}
\end{footnotesize}
\end{threeparttable}
\end{adjustwidth}
\end{table}

Tables \ref{tab:sim_250_3} and \ref{tab:sim_250_4} consider large monotonicity violations and ``weak cells.'' It remains the case that IJIVE, UJIVE, and FEJIV are nearly unbiased whenever the average value of $\widetilde{F}$ is large enough. This includes the ``weak cells'' design in Table \ref{tab:sim_250_4}, which underscores the notion that the instrument can be weak in some groups as long as it is sufficiently strong in others.\footnote{Intuitively, if $\pi(x)=0$ when $X=x$, $\tau(x)$ is not identified. However, because $\tau_{\mathrm{LATE}} = \frac{\e \left[ \pi(X) \cdot \tau(X) \right]}{\e \left[ \pi(X) \right]}$, the weight on $\tau(x)$ in $\tau_{\mathrm{LATE}}$ would have been zero anyway, and analogously for the estimands in Theorem \ref{the:ai1995}, Theorem \ref{the:justid_cm}, and Corollary \ref{cor:justid_um}. That is, as long as the overall instrument strength is sufficient \citep[cf.][]{MS2022}, it does not matter that some conditional LATEs cannot be well estimated due to a conditional-on-$X$ weak IV problem, because those conditional LATEs are irrelevant for the target estimand.} On the other hand, unlike in Table \ref{tab:sim_250_2}, IV estimation of the noninteracted specification is not only biased in Tables \ref{tab:sim_250_3} and \ref{tab:sim_250_4}, but also noisy, which leads to very high values of MSE\@.

The remaining simulation results, for ``weak'' IV with $|\mathcal{X}| = 20$ and for ``strong'' IV with both values of $|\mathcal{X}|$, are reported in Tables \ref{tab:sim_250_5}--\ref{tab:sim_20_8} in Appendix \ref{app:simul}\@. The bottom line is still that \cite{MS2022}'s pretest does a great job differentiating between cases where IJIVE, UJIVE, and FEJIV perform well or very well, and cases where all estimators of the interacted specification perform badly. Roughly speaking, values of $\widetilde{F}$ exceeding the recommended cutoff of 4.14 are associated with low bias, even when, with $|\mathcal{X}| = 250$ and $N=3{,}000$, there are only 12 units in each group; values of $\widetilde{F}$ exceeding 10--15 are associated with negligible or no bias, at least in the data-generating process under consideration.

Other estimators are clearly not competitive with IJIVE, UJIVE, and FEJIV\@. When there are violations of monotonicity, the usual application of IV is biased and often unstable. 2SLS estimation of the interacted specification is generally biased, as expected. MB2SLS is usually dominated by IJIVE, UJIVE, and FEJIV, especially on bias. JIVE is generally biased and unstable.

To be clear, the purpose of this simulation study is not to claim that \citetalias{AI1995}'s specification can be estimated without bias in a typical application of IV methods. Instead, the simulations show that IJIVE, UJIVE, and FEJIV estimation of \citetalias{AI1995}'s specification is reliable \textit{if} the instruments are jointly strong enough, which can be verified using \cite{MS2022}'s pretest. Future research should examine whether the number of instruments in \citetalias{AI1995}'s specification could be reduced using appropriate regularization techniques, perhaps a modification of the existing approaches in \cite{CHS2015a,CHS2015b} and \cite{Wiemann2024}.

\section{Empirical Applications}
\label{sec:empirics}

The results so far underscore the importance of using the interacted specification when weak monotonicity is plausible but strong monotonicity is not. In this section I present evidence of violations of strong monotonicity and first-stage homogeneity in a sample of recent applications of IV methods. Then, I revisit a study of the effects of pretrial detention on case outcomes in Philadelphia, where violations of strong monotonicity are particularly evident \citep{Stevenson2018}.

\subsection{Review of Applications of Instrumental Variables}
\label{sec:young}

In what follows, I use a sample of 1,309 instrumental variables regressions previously analyzed by \cite{Young2022}, which corresponds to the universe of IV estimates reported in the main text of 30 papers published in journals of the American Economic Association between 2006 and 2015.\footnote{\cite{Young2022}'s goal was to cover the universe of replicable IV applications in this period subject to a small number of additional inclusion criteria reported in his paper.} After dropping specifications with multiple instruments, without additional covariates, or based on panel data, I obtain my final sample of 988 regressions in 25 papers.\footnote{Because \cite{Young2022} only considered papers with replication code in Stata, I define ``specifications based on panel data'' as those using Stata's \texttt{xtivreg} or \texttt{xtivreg2} commands in the original replication package. The number of applicable regressions in several papers would decrease substantially if we eliminated not only duplicate IV regressions---which \cite{Young2022} already did---but also duplicate first stages. However, my preliminary attempt to do so did not meaningfully change any of the results reported in this section.} The number of regressions per paper is highly uneven in this sample, with the mean equal to $988/25 = 39.52$ and the quartiles equal to 8, 14, and 40. The list of papers under consideration is provided in Appendix \ref{app:young}\@.

Given the inclusion criteria above, every specification in my sample is based on a linear first-stage regression of $D$ on $Z$ and $X$, without any interactions between $Z$ and $X$\@. In my first exercise, I implicitly include these interactions by means of separate regressions of $D$ on $X$ given $Z=1$ and $Z=0$.\footnote{If the original treatment or instrument are not binary, I replace them with indicators for whether these variables are above their medians. I demonstrate robustness to other binarizations in Tables \ref{tab:apps_neg} and \ref{tab:apps_het} in Appendix \ref{app:young}\@.} This simple approach allows me to estimate the conditional first stage at every value of $X$ as the difference in conditional means, $\hat{\omega}(x) = \hat{\e} \left[ D \mid Z=1, X=x \right] - \hat{\e} \left[ D \mid Z=0, X=x \right]$. Subsequently, I report the fraction of these estimates that are opposite in sign (``negative'') to the estimate in the original first stage, which is equivalent to the fraction of observations with negative weights in the usual application of IV (cf.~Theorem \ref{the:justid_cm}). This is analogous to the recommendation of \cite{dCDH2020} to report the fraction of units with negative weights in two-way fixed effects regressions. Similarly, \cite{Semenova2025} reports the fraction of observations with negative predictions in a sample selection context related to mine.

Panel A of Table \ref{tab:apps_1_1} indicates that negative first stages are a common occurrence in recent applications of IV\@. The average fraction of observations with a negative first stage is 21.8\% when using the linear probability model (LPM) to estimate the conditional means in $\omega(x)$ and 17.6\% when using the probit model. After weighting by the inverse of the number of applicable regressions associated with a given paper, these averages increase to 28.5\% and 28.0\%, respectively, giving the average of the within-paper averages.

It may be the case that a portion of the estimated negative first stages is due to noise. However, the regressions in my sample are usually not saturated, which means that the formal test of violations of monotonicity in Remark \ref{rem:montest} is not appropriate. Instead, in my second exercise, I explicitly add interaction terms to each original (linear) first stage and test whether the corresponding coefficients are jointly equal to zero. Under the alternative, the true first stage is heterogeneous, which is a necessary condition for strong monotonicity being false but weak monotonicity being true.

\begin{table}[!t]
\begin{adjustwidth}{-1in}{-1in}
\centering
\begin{threeparttable}
\caption{Main Results on Negative First Stages and First-Stage Heterogeneity\label{tab:apps_1_1}}
\begin{tabular}{>{\centering\arraybackslash}m{7.5cm} >{\centering\arraybackslash}m{2.25cm} >{\centering\arraybackslash}m{2.25cm}}
\hline\hline
    A. Negative First Stages & LPM   & Probit \\
\hline
    Average Share & 0.218 & 0.176 \\
    Weighted Average Share & 0.285 & 0.280 \\
          &       &  \\
\hline
    B. First-Stage Heterogeneity & LPM   & Probit \\
\hline
    Rejected Papers & 22/25 & 19/21 \\
    Average Share of Rejections & 0.715 & 0.749 \\
\hline
\end{tabular}
\begin{footnotesize}
\begin{tablenotes}[flushleft]
\item \textit{Notes:} Panel A reports summary statistics on the fraction of observations for which $\hat{\e} \left[ D \mid Z=1, X=x \right] - \hat{\e} \left[ D \mid Z=0, X=x \right]$ is negative. ``Average Share'' treats every applicable regression equally. ``Weighted Average Share'' weights by the inverse of the number of applicable regressions associated with a given paper. Panel B reports results of Wald tests that the coefficients on the interaction terms in regressions of $D$ on $Z$, $X$, and $ZX$ are jointly equal to zero. ``Rejected Papers'' reports the number of papers for which the Bonferroni $p$-value is less than or equal to 0.05. ``Average Share of Rejections'' reports the average share (across papers) of regressions associated with a given paper for which the corresponding Holm $p$-value is less than or equal to 0.05. $D$ and $Z$ are defined as either the original endogenous explanatory variable and instrument (if they are binary) or indicators for whether these variables are above their medians, subject to a normalization that $Z$ is always associated with a positive estimated coefficient in the linear first stage. Sampling weights and clustered standard errors are used in line with the original papers. Paper-specific results are reported in Table \ref{tab:apps_1_2} in Appendix \ref{app:young}\@.
\end{tablenotes}
\end{footnotesize}
\end{threeparttable}
\end{adjustwidth}
\end{table}

Panel B of Table \ref{tab:apps_1_1} reports the results of this exercise. Using the Bonferroni procedure to account for multiple hypothesis testing separately for each paper, I conclude that 22 of 25 papers have at least one first stage that is heterogeneous. Using the Holm correction, I reject an average of 71.5\% of homogeneous first stages per paper. The last column demonstrates that these conclusions are robust to using the probit instead of the linear probability model (LPM)\@.\footnote{The smaller number of papers under consideration when using the probit model reflects convergence and other estimation problems in the missing specifications.}

\subsection{Reanalysis of \cite{Stevenson2018}}
\label{sec:stevenson}

Now, I turn to a reanalysis of \cite{Stevenson2018}'s study of the effects of pretrial detention on case outcomes. In this application, recently reanalyzed by \cite{Cunningham2021}, \cite{CHMW2024}, and \cite{MT2024}, violations of strong monotonicity are evident, which I will be able to formally demonstrate.

\begin{table}[!t]
\begin{adjustwidth}{-1in}{-1in}
\centering
\begin{threeparttable}
\caption{Eight Judges in \cite{Stevenson2018}\label{tab:stevenson_1}}
\begin{tabular}{>{\centering\arraybackslash}m{2.75cm} >{\centering\arraybackslash}m{1.75cm} >{\centering\arraybackslash}m{1.75cm}}
\hline\hline
     & $N$ & Detention Rate \\
\hline
    Judge A & 21,523 & 0.402 \\
    Judge B & 13,087 & 0.432 \\
    Judge C & 56,585 & 0.395 \\
    Judge D & 33,690 & 0.413 \\
    Judge E & 55,038 & 0.432 \\
    Judge F & 41,475 & 0.413 \\
    Judge G & 56,301 & 0.398 \\
    Judge H & 54,272 & 0.418 \\
\hline
\end{tabular}
\begin{footnotesize}
\begin{tablenotes}[flushleft]
\item \textit{Notes:} The data are \cite{Stevenson2018}'s sample of 331,971 arrests in Philadelphia. $N$ is the number of cases heard by a given judge. ``Detention Rate'' is the proportion of cases heard by a given judge such that the defendant is subsequently detained pretrial.
\end{tablenotes}
\end{footnotesize}
\end{threeparttable}
\end{adjustwidth}
\end{table}

The data are based on the Philadelphia court records and cover 331,971 arrests between 2006 and 2013. The ``treatment'' of interest is pretrial detention or, in other words, whether the defendant was incarcerated in the period between their arrest and disposition; the purpose of such detention is that they appear in court and do not commit another crime. The empirical question is whether pretrial detention has a causal effect on case outcomes, such as conviction and incarceration length. Naturally, pretrial detention is endogenous, and \cite{Stevenson2018}'s identification strategy is based on random assignment of bail magistrates (judges) to cases. These judges have broad authority to set bail---the amount required for pretrial release---at a level they choose. Thus, being assigned a strict judge makes the defendant less likely to be able to pay bail and more likely to be detained.

In Philadelphia, bail hearings usually last one or two minutes, which made it possible for only eight judges to hear all the cases in \cite{Stevenson2018}'s data. Table \ref{tab:stevenson_1} reports the number of cases and the detention rate for each judge. The magistrate I refer to as ``Judge C'' is the most lenient, with a relatively low detention rate of 0.395. In what follows---unlike \cite{Stevenson2018}, who uses the full set of judge indicators as instruments---I focus on a single instrument defined as whether a given case was heard by Judge C\@. A simple regression of pretrial detention on the ``Judge C'' dummy reveals a first stage of --0.0195 with a standard error of 0.0023.

In the present context, strong monotonicity requires that every defendant detained by Judge C would also have been detained by other judges. However, this condition seems implausible, with the likely dimensions of monotonicity violations including the offense type \citep{Stevenson2018} and the defendant's race \citep{ABM2012}. As in \cite{Stevenson2018}, I focus on the seventeen most common offense types.\footnote{These include drug possession, drug sale, aggravated assault, robbery, first offense DUI, simple assault, drug purchase, burglary, shoplifting, theft, marijuana possession, murder, motor vehicle theft, prostitution, third-degree felony firearm possession, second-degree felony firearm possession, and vandalism.} I also consider three racial categories: Black, White, and other. The offense types are not mutually exclusive, which means that, in principle, the sample could be divided into $3 \cdot 2^{17}$ groups based on the defendant's race and the offense type. However, most of these groups are empty, and I also drop nonempty groups with fewer than three cases heard by Judge C or not heard by Judge C\@. As a result, for this specification, the final sample consists of 431 groups and 327,560 cases.

With such a saturated specification, a formal test of violations of strong monotonicity is straightforward, as discussed in Remark \ref{rem:montest}. Given that Judge C is more lenient than others, the overall first stage is negative---not positive, as assumed previously---and the null hypothesis requires that all the group-specific first stages are also non-positive. In other words, the null can be written as
\begin{equation}
H_0 : \; \omega \leq 0
\end{equation}
while the test statistic equals
\begin{equation}
\label{eq:supt2}
T = \max_{1 \leq k \leq K} \frac{\hat{\omega}_k}{\hat{\sigma}_{\hat{\omega}_k}}.
\end{equation}
When I implement this test, I obtain a test statistic of 5.637 and a $p$-value of 3.7e-06, despite accounting for multiple hypothesis testing.\footnote{In this application, the approach of \cite{CCK2019} produces almost identical critical values as the Bonferroni procedure.} In this application, strong monotonicity is clearly rejected. Further details on the group-specific impact of Judge C on pretrial detention are provided in Table \ref{tab:stevenson_2}. Because presenting estimates for 431 groups is impractical, I restrict my attention to twenty groups with the largest (most positive) and ten groups with the smallest (most negative) $z$ statistics. For each group, I report the number of cases, the conditional first stage and its standard error, and the corresponding Holm $p$-value. At any conventional significance level, we can reject that the first stage is non-positive in two groups: defendants charged with burglary and vandalism who are neither Black nor White and Black defendants charged with robbery, simple assault, and theft. In general, a common feature of many of the groups with the largest $z$ statistics is a combination of being charged with a property crime (\textit{e.g.}, theft or burglary) and a violent crime (\textit{e.g.}, simple assault or aggravated assault). In fact, seven of these groups are charged with robbery, which is simultaneously a violent crime and a crime against property.\footnote{This is consistent with \citet[][p.~525]{Stevenson2018}'s account that ``[t]he magistrate that is most lenient overall is actually strictest when it comes to robbery.'' However, the test discussed in Remark \ref{rem:montest} and the results in Table \ref{tab:stevenson_2} are otherwise different from the analysis in \cite{Stevenson2018}.} Many of these groups comprise of defendants who are neither Black nor White. On the other hand, the groups with the smallest $z$ statistics are universally Black or White, and charged with nonviolent crimes.

\renewcommand{\arraystretch}{1.3}

\begin{table}[!p]
\begin{adjustwidth}{-1in}{-1in}
\centering
\begin{threeparttable}
\caption{Conditional First Stages and Violations of Strong Monotonicity in \cite{Stevenson2018}\label{tab:stevenson_2}}
\begin{footnotesize}
\begin{tabular}{>{\centering\arraybackslash}m{10.25cm} >{\centering\arraybackslash}m{1.5cm} >{\centering\arraybackslash}m{1.5cm} >{\centering\arraybackslash}m{1.5cm} >{\centering\arraybackslash}m{1.5cm} >{\centering\arraybackslash}m{1.5cm}}
\hline\hline
    Offense Type & Race & $N$ & $\hat{\omega}_k$ & $\hat{\sigma}_{\hat{\omega}_k}$ & Holm $p$-value \\
\hline
    burglary and vandalism & other & 200   & 0.432*** & 0.077 & 3.73e-06 \\
    robbery, simple assault, and theft & Black & 5,033 & 0.058*** & 0.014 & 0.00869 \\
    aggravated assault, robbery, simple assault, and theft & other & 279   & 0.215** & 0.085 & 1 \\
    first offense DUI and marijuana possession & other & 135   & 0.143** & 0.057 & 1 \\
    drug possession, robbery, simple assault, and theft & Black & 93    & 0.253** & 0.104 & 1 \\
    drug purchase & White & 16    & 0.333** & 0.140 & 1 \\
    drug purchase and marijuana possession & Black & 111   & 0.218** & 0.092 & 1 \\
    robbery, simple assault, theft, and third-degree felony firearm possession & other & 221   & 0.163** & 0.068 & 1 \\
    aggravated assault, drug possession, drug sale, and simple assault & Black & 76    & 0.297** & 0.134 & 1 \\
    aggravated assault, simple assault, and theft & other & 30    & 0.432** & 0.195 & 1 \\
    aggravated assault and simple assault & Black & 14,262 & 0.024** & 0.011 & 1 \\
    aggravated assault, robbery, simple assault, and third-degree felony firearm possession & other & 11    & 0.500** & 0.247 & 1 \\
    theft and vandalism & other & 236   & 0.151** & 0.076 & 1 \\
    burglary, theft, and vandalism & other & 406   & 0.123** & 0.062 & 1 \\
    aggravated assault, first offense DUI, and simple assault & Black & 94    & 0.222* & 0.114 & 1 \\
    burglary, robbery, theft, and third-degree felony firearm possession & Black & 216   & 0.062* & 0.033 & 1 \\
    robbery, simple assault, and theft & other & 865   & 0.079* & 0.043 & 1 \\
    third-degree felony firearm possession & White & 416   & 0.116* & 0.064 & 1 \\
    shoplifting and vandalism & Black & 41    & 0.302* & 0.169 & 1 \\
    burglary and theft & other & 342   & 0.119* & 0.067 & 1 \\
    \ldots & \ldots & \ldots & \ldots & \ldots & \ldots \\
    drug possession and marijuana possession & Black & 8,599 & --0.049*** & 0.011 & 1 \\
    shoplifting & White & 4,132 & --0.088*** & 0.020 & 1 \\
    motor vehicle theft & White & 890   & --0.193*** & 0.041 & 1 \\
    drug possession and drug purchase & Black & 6,885 & --0.070*** & 0.014 & 1 \\
    motor vehicle theft & Black & 2,183 & --0.138*** & 0.026 & 1 \\
    drug possession & White & 10,035 & --0.052*** & 0.010 & 1 \\
    drug possession and drug purchase & White & 7,692 & --0.061*** & 0.011 & 1 \\
    prostitution & Black & 2,967 & --0.120*** & 0.022 & 1 \\
    theft & Black & 5,886 & --0.098*** & 0.017 & 1 \\
    shoplifting & Black & 8,065 & --0.115*** & 0.014 & 1 \\
\hline
\end{tabular}
\begin{tablenotes}[flushleft]
\item \textit{Notes:} The data are \cite{Stevenson2018}'s sample of 331,971 arrests in Philadelphia. The first two columns identify one of the 431 groups based on the offense type and the defendant's race. $N$ is the number of cases in a given group. $\hat{\omega}_k$ is the conditional first stage, that is, the group-specific effect of Judge C on pretrial detention. $\hat{\sigma}_{\hat{\omega}_k}$ is the first-stage standard error. Entries in the table are sorted in descending order of $z = \hat{\omega}_k / \hat{\sigma}_{\hat{\omega}_k}$ and are restricted to twenty groups with the largest (most positive) and ten groups with the smallest (most negative) $z$ statistics. Holm $p$-value equals $\min(1,p^*)$, where $p^*$ is the product of the group-specific $p$-value for a one-sided test, based on $\hat{\omega}_k / \hat{\sigma}_{\hat{\omega}_k}$, and $r_k+1$, where $r_k$ is the number of group-specific $p$-values greater than that for a given group.
\item *Statistically different from zero at the 10\% level; **at the 5\% level; ***at the 1\% level.
\end{tablenotes}
\end{footnotesize}
\end{threeparttable}
\end{adjustwidth}
\end{table}

\renewcommand{\arraystretch}{1}

Because strong monotonicity is rejected, the noninteracted specification cannot be used to estimate a convex combination of conditional LATEs (cf.~Theorem \ref{the:justid_cm}). However, if weak monotonicity is plausible, the interacted specification will be appropriate, at least as long as the interacted instruments are sufficiently strong. In the present context, weak monotonicity seems quite sensible. Given that bail hearings in Philadelphia are extremely short, it is unlikely that more than a handful of factors---such as the offense type and the demographic characteristics of the defendant---could determine the amount of bail and the resulting likelihood of detention.

To incorporate additional factors into the analysis, I also consider two alternative specifications. First, I define the groups based on the offense type, the defendant's race, and their gender (male or female). In theory, the number of groups could be as large as $3 \cdot 2^{18}$ in this specification, but only 563 groups remain after I drop those that are empty or otherwise too small---requiring, as above, that there are at least three observations for every $(G,Z)$ combination. Second, I define the groups based on the offense type, the defendant's race and gender, and three time periods considered by \cite{Stevenson2018}. The relevance of these specific time periods---divided by February 23, 2009 and February 23, 2011---results from concurring changes in the composition of magistrates other than Judge C\@. This sets the maximum number of groups in this specification at $3^2 \cdot 2^{18}$; in practice, the number of groups that are nonempty and sufficiently large is 981.

Table \ref{tab:stevenson_3} reports the main results of my analysis. In panels C and D, for each of the three specifications described above, I report the Bonferroni/\cite{CCK2019} $p$-value for the test of violations of strong monotonicity as well as \cite{MS2022}'s test statistic, $\widetilde{F}$, for the test of weak identification. The test results leave little doubt that strong monotonicity is violated while identification is strong. The $p$-values for the former test never exceed 0.0015, despite accounting for simultaneous testing of up to 981 hypotheses.\footnote{In the second and third specifications, the largest $z$ statistic is obtained in very small groups, which makes the normal approximation questionable. However, the rejection of strong monotonicity remains solid. The smallest Holm $p$-values in groups with at least 100 cases are 0.0109 and 0.0105 in the second and third specifications, respectively. The corresponding smallest Holm $p$-values in groups with at least 500 cases are 0.0408 and 0.0105. Note that these $p$-values are conservative, because they implicitly penalize hypothesis testing in groups smaller than 100 or 500 cases, even though such groups are ignored in this context.} The values of $\widetilde{F}$ range between 19.32 and 21.56. If the simulations in Section \ref{sec:simul} are any guide, we should expect negligible bias when estimating the interacted specification, at least when using the jackknife-type estimators such as IJIVE, UJIVE, and FEJIV\@.

Panels A and B of Table \ref{tab:stevenson_3} report OLS, IV, 2SLS, IJIVE, UJIVE, and FEJIV estimates of the effects of pretrial detention on conviction and incarceration length. The noninteracted specification, marked as ``IV,'' suggests that pretrial detention leads to a 17--19 p.p.~increase in the likelihood of being convicted and an increase in incarceration length of 670--720 days. Such effects would have been substantial, but the validity of these estimates is questionable given the clear rejection of strong monotonicity in this application. When we turn to the interacted specification, the estimates become smaller. The effects on conviction are closer to zero---in the range of 4 to 15 p.p.---but often remain significant.\footnote{In the case of FEJIV, I report the standard errors derived by \cite{CSW2023}. Practitioners should also consider a recent alternative proposed by \cite{BN2024}, which explicitly accounts for treatment effect heterogeneity.} The effects on incarceration length are much smaller than in the noninteracted specification and suggestive of an effect of 50--160 days. These estimates are also usually not significantly different from zero, except for 2SLS\@. For both outcomes and each specification, the IJIVE, UJIVE, and FEJIV estimates are practically indistinguishable from each other but also clearly different from the 2SLS estimate.

\begin{table}[!tb]
\begin{adjustwidth}{-1in}{-1in}
\centering
\begin{threeparttable}
\caption{Causal Effects of Pretrial Detention on Conviction and Incarceration Length\label{tab:stevenson_3}}
\begin{small}
\begin{tabular}{cc >{\centering\arraybackslash}m{1.7cm} >{\centering\arraybackslash}m{1.7cm} c >{\centering\arraybackslash}m{1.7cm} >{\centering\arraybackslash}m{1.7cm} c >{\centering\arraybackslash}m{1.7cm} >{\centering\arraybackslash}m{1.7cm}}
\hline\hline
          &       & \multicolumn{2}{c}{Specification \#1} &       & \multicolumn{2}{c}{Specification \#2} &       & \multicolumn{2}{c}{Specification \#3} \\
\hline
    A. Effects on Conviction &       & $\hat{\beta}$ & $\hat{\sigma}_{\hat{\beta}}$ &       & $\hat{\beta}$ & $\hat{\sigma}_{\hat{\beta}}$ &       & $\hat{\beta}$ & $\hat{\sigma}_{\hat{\beta}}$ \\
\hline
    OLS   &       & 0.0591*** & 0.0019 &       & 0.0567*** & 0.0019 &       & 0.0530*** & 0.0019 \\
    IV    &       & 0.1852* & 0.1070 &       & 0.1920* & 0.1022 &       & 0.1704 & 0.1039 \\
    2SLS  &       & 0.1116*** & 0.0360 &       & 0.1193*** & 0.0325 &       & 0.0610** & 0.0271 \\
    IJIVE &       & 0.1283** & 0.0585 &       & 0.1433*** & 0.0540 &       & 0.0429 & 0.0559 \\
    UJIVE &       & 0.1304** & 0.0616 &       & 0.1466** & 0.0582 &       & 0.0419 & 0.0589 \\
    FEJIV &       & 0.1292** & 0.0604 &       & 0.1451** & 0.0569 &       & 0.0417 & 0.0574 \\
          &       &       &       &       &       &       &       &       &  \\
\hline
    B. Effects on Incarceration Length &       & $\hat{\beta}$ & $\hat{\sigma}_{\hat{\beta}}$ &       & $\hat{\beta}$ & $\hat{\sigma}_{\hat{\beta}}$ &       & $\hat{\beta}$ & $\hat{\sigma}_{\hat{\beta}}$ \\
\hline
    OLS   &       & 184*** & 3     &       & 176*** & 2     &       & 173*** & 3 \\
    IV    &       & 689*** & 249   &       & 723*** & 237   &       & 666*** & 233 \\
    2SLS  &       & 158*** & 47    &       & 159*** & 41    &       & 130*** & 43 \\
    IJIVE &       & 95    & 76    &       & 124*  & 68    &       & 58    & 79 \\
    UJIVE &       & 93    & 83    &       & 123   & 75    &       & 56    & 99 \\
    FEJIV &       & 92    & 78    &       & 122*  & 70    &       & 51    & 91 \\
          &       &       &       &       &       &       &       &       &  \\
\hline
    C. Test of Violations of Monotonicity &       &       &       &       &       &       &       &       &  \\
\hline
    $p$-value &       & \multicolumn{2}{c}{3.73e-06} &       & \multicolumn{2}{c}{0.00133} &       & \multicolumn{2}{c}{1.83e-08} \\
          &       &       &       &       &       &       &       &       &  \\
\hline
    D. Pretest for Weak Identification &       &       &       &       &       &       &       &       &  \\
\hline
    $\widetilde{F}$ &       & \multicolumn{2}{c}{21.38} &       & \multicolumn{2}{c}{21.56} &       & \multicolumn{2}{c}{19.32} \\
          &       &       &       &       &       &       &       &       &  \\
\hline
    Number of Groups &       & \multicolumn{2}{c}{431} &       & \multicolumn{2}{c}{563} &       & \multicolumn{2}{c}{981} \\
    Number of Observations &       & \multicolumn{2}{c}{327,560} &       & \multicolumn{2}{c}{325,915} &       & \multicolumn{2}{c}{319,573} \\
\hline
\end{tabular}
\end{small}
\begin{footnotesize}
\begin{tablenotes}[flushleft]
\item \textit{Notes:} The data are \cite{Stevenson2018}'s sample of 331,971 arrests in Philadelphia. The outcomes are conviction (Panel A) or incarceration length (Panel B), defined as the maximum days of an incarceration sentence. The treatment is pretrial detention. The instrument is whether a given case was heard by Judge C\@. Each specification is based on a division of the sample into a number of mutually exclusive groups, with a separate group for each combination of values of selected variables. Specification \#1 uses the offense type and race (Black, White, or other) of the defendant. Specification \#2 uses the offense type, race, and gender (male or female) of the defendant. Specification \#3 uses the offense type, race and gender of the defendant, and three time periods considered by \cite{Stevenson2018}. Groups with fewer than three observations in either $(G,Z)$ combination are dropped. ``OLS'' is the OLS estimator in the regression of the outcome on the treatment indicator and group indicators. ``IV'' is the IV estimator in the noninteracted specification. The remaining estimators are based on the interacted specification and are described in Section \ref{sec:manyiv}. The test of violations of monotonicity is described in Remark \ref{rem:montest} and reports the Bonferroni/\cite{CCK2019} $p$-values. The pretest for weak identification follows \cite{MS2022} and reports their test statistic, $\widetilde{F}$\@. The cutoff for this test is 4.14. See also the Stata implementation in \cite{Sun2023}.
\item *Statistically different from zero at the 10\% level; **at the 5\% level; ***at the 1\% level.
\end{tablenotes}
\end{footnotesize}
\end{threeparttable}
\end{adjustwidth}
\end{table}

My conclusions are generally in line with \cite{Stevenson2018}, whose paper includes a relatively rare recent example of using specifications with interacted instruments (cf.~Remark \ref{rem:lackofoverid}), although not of \citetalias{AI1995}'s interacted specification, which is implemented here. I also provide additional results in Appendix \ref{app:stevenson}\@. Table \ref{tab:stevenson_extra} reports MB2SLS and JIVE estimates of the effects of pretrial detention. While the MB2SLS estimates are largely similar to the results in Table \ref{tab:stevenson_3}, the JIVE estimates are noisy and appear unreliable. Table \ref{tab:stevenson_boot1} reports the results of a bootstrap test for comparisons between the IV estimates in the noninteracted specification and various estimates of the interacted specification.\footnote{I perform this test for 2SLS, MB2SLS, JIVE, and UJIVE, but not IJIVE and FEJIV, because the latter estimators are very computationally demanding in the specifications that I consider, at least in my implementation.} At the 5\% level, I nearly always reject the null that the estimands are the same in the case of incarceration length but not conviction. Finally, Table \ref{tab:stevenson_boot2} reports the results of a similar bootstrap test for comparisons between 2SLS and other estimators of the interacted specification. These differences are often highly statistically significant, which reaffirms the importance of correcting for the many instrument bias.

\section{Conclusion}
\label{sec:conclusion}

In this paper I studied the interpretation of linear IV and 2SLS estimands when both the treatment and the instrument are binary, and when additional covariates are required for identification. Using the LATE framework of \cite{IA1994} and \cite{AIR1996}, I argued that the common practice of interpreting standard IV estimands as a convex combination of conditional LATEs, or even as the (unconditional) local average treatment effect, is substantially more problematic than previously thought. I showed that the interpretation of the usual application of IV, which limits the effects of the instrument in the reduced-form and first-stage regressions to be homogeneous, hinges critically on the specific variant of the monotonicity assumption that the researcher is willing to entertain. Under ``weak monotonicity,'' some of the IV weights may be negative and the IV estimand may no longer be interpretable as a causal effect.

What should applied researchers do in practice? In this paper I argued that it might be worthwhile to revisit the interacted specification of \cite{AI1995}, which is guaranteed to eliminate negative weights under the same assumptions that are problematic for the usual application of IV\@. Specifications with many interacted instruments were used in influential papers by \cite{Angrist1990} and \cite{AK1991} but appear to have been largely abandoned in subsequent work out of concern for the many instrument bias. Unsurprisingly, however, the modern tools to estimate such specifications are substantially better than in the 1990s, as I also demonstrate in an extensive simulation study. A pretest for weak identification developed by \cite{MS2022} can be used to determine whether consistent estimation of the interacted specification is possible. When the pretest rejects, several jackknife-type estimators can be used, including the FEJIV estimator of \cite{CSW2023}, which I also implement in the companion MATLAB, R, and Stata packages, \texttt{fejiv}.

There are at least two important situations when this recommendation will not be satisfactory. First, in some applications in which strong monotonicity is rejected, weak monotonicity will be implausible, too. If this is the case, it may be worthwhile to instead consider the partial identification approach of \cite{Noack2021}, which evaluates the sensitivity of what can be learned about the local average treatment effect under violations of (weak) monotonicity. Second, a convex combination of conditional LATEs, which \cite{AI1995}'s specification is guaranteed to produce under weak monotonicity (and the usual application of IV under strong monotonicity), may be considered an imperfect substitute for the (unconditional) local average treatment effect.\footnote{A similar argument is made by \cite{CGBS2024} in the context of difference-in-differences designs.} If this is the case, there are many existing estimators that are consistent for the LATE under strong monotonicity (see, \textit{e.g.}, \citealp{SUW_kappa}, and the references therein). An important avenue for future research is to develop estimators of the LATE that are also robust to weak monotonicity.

\singlespacing

\setlength\bibsep{0pt}
\bibliographystyle{ECA_ff}
\bibliography{LATE_references}

\begin{appendices}

\renewcommand{\contentsname}{Appendices}
\renewcommand{\refname}{Appendix References}
\renewcommand{\thesection}{\Alph{section}}
\renewcommand{\thesubsection}{\Alph{section}\arabic{subsection}}

\onehalfspacing

\section{Proofs}
\label{app:proofs}

\setcounter{equation}{0}
\renewcommand{\theequation}{A.\arabic{equation}}

\paragraph{Proof of Theorem \ref{the:ai1995}.}

Lemma \ref{lem:ai1995} states that $\beta_{\mathrm{2SLS}} = \frac{\e \left[ \sigma^2(X) \cdot \tau(X) \right]}{\e \left[ \sigma^2(X) \right]}$. It remains to show that $\sigma^2(X) = \left[ \pi(X) \right] ^2 \cdot \var \left[ Z \mid X \right]$. Indeed, it follows from the definition of $\sigma^2(X)$, equation (\ref{eq:fs}), and iterated expectations that $\sigma^2(X) = \left[ \omega(X) \right] ^2 \cdot \var \left[ Z \mid X \right]$. Then, it follows from Lemma \ref{lem:identify} that $\sigma^2(X) = \left[ \pi(X) \right] ^2 \cdot \var \left[ Z \mid X \right]$ because $\left[ \omega(X) \right] ^2 = \left[ \pi(X) \right] ^2$ under Assumptions \ref{ass:iv} and \ref{ass:cm}\@.

\bigskip

\paragraph{Proof of Theorem \ref{the:justid_cm}.}

Let $R$ and $T$ be generic notation for two random variables, where $T$ is binary and $R$ is arbitrarily discrete or continuous. The following lemma, due to \cite{Angrist1998}, will be useful for what follows.

\begin{lemma}[\citealp{Angrist1998}]
\label{lem:angrist1998}
Suppose that $\e \left[ T \mid X \right]$ is linear in $X$\@. Then, $\xi$, the coefficient on $T$ in the linear projection of $R$ on $T$ and $X$ can be written as
\begin{equation*}
\xi = \frac{\e \left[ \var \left[ T \mid X \right] \cdot \xi(X) \right]}{\e \left[ \var \left[ T \mid X \right] \right]},
\end{equation*}
where $\xi(X) = \e \left[ R \mid T=1, X \right] - \e \left[ R \mid T=0, X \right]$.
\end{lemma}

\noindent
Recall that $\beta_{\mathrm{IV}}$ is equal to the ratio of the reduced-form and first-stage coefficients on $Z$\@. It follows that we can apply Lemma \ref{lem:angrist1998} separately to these two coefficients, and thereby obtain the following expression for the estimand of interest:
\begin{equation}
\beta_{\mathrm{IV}} = \frac{\frac{\e \left[ \var \left[ Z \mid X \right] \cdot \phi(X) \right]}{\e \left[ \var \left[ Z \mid X \right] \right]}}{\frac{\e \left[ \var \left[ Z \mid X \right] \cdot \omega(X) \right]}{\e \left[ \var \left[ Z \mid X \right] \right]}},
\end{equation}
where
\begin{equation}
\label{eq:crf}
\phi(x) = \e \left[ Y \mid Z=1, X=x \right] - \e \left[ Y \mid Z=0, X=x \right]
\end{equation}
is the conditional reduced-form slope coefficient and $\omega(x)$ is as defined in equation (\ref{eq:cfs}). Upon rearrangement, we obtain
\begingroup
\jot=8pt
\begin{eqnarray}
\beta_{\mathrm{IV}} &=& \frac{\e \left[ \var \left[ Z \mid X \right] \cdot \phi(X) \right]}{\e \left[ \var \left[ Z \mid X \right] \cdot \omega(X) \right]}
\nonumber\\
\label{eq:ivdesc}
&=& \frac{\e \left[ \var \left[ Z \mid X \right] \cdot \omega(X) \cdot \beta(X) \right]}{\e \left[ \var \left[ Z \mid X \right] \cdot \omega(X) \right]},
\end{eqnarray}
\endgroup
where the second equality uses the definition of $\beta(x)$ in equation (\ref{eq:wald}). See also \cite{Walters2018} for a similar argument. Finally, we know from Lemma \ref{lem:identify} that $\beta(x) = \tau(x)$ and $\omega(x) = c(x) \cdot \pi(x)$ under Assumptions \ref{ass:iv} and \ref{ass:cm}\@. This completes the proof because $\beta_{\mathrm{IV}}$ can now be written as
\begin{equation}
\beta_{\mathrm{IV}} = \frac{\e \left[ c(X) \cdot \pi(X) \cdot \var \left[ Z \mid X \right] \cdot \tau(X) \right]}{\e \left[ c(X) \cdot \pi(X) \cdot \var \left[ Z \mid X \right] \right]}.
\end{equation}

\bigskip

\paragraph{Alternative Proof of Theorem \ref{the:justid_cm}.}

The following proof of Theorem \ref{the:justid_cm} uses \cite{Kolesar2013}'s result in Remark \ref{rem:kolesar}. Let us begin by restating the representation of two-step IV estimands in equation (\ref{eq:kolesar}):
\begin{equation}
\label{eq:kolesar2}
\beta_{\mathrm{TSIV}} = \int \sum_{j=1}^{J_x-1} \frac{\theta_j(x)}{\int \sum_{j=1}^{J_x-1} \theta_j(x) \, \mathrm{d} F^X(x)} \, \tau(p_{j,x};x) \, \mathrm{d} F^X(x).
\end{equation}
The notation is the same as in Remark \ref{rem:kolesar}. The representation in equation (\ref{eq:kolesar2}) is appropriate for any two-step IV estimand (\textit{e.g.}, 2SLS) which uses $Z_{\mathrm{G}} = z_{\mathrm{G}}(X,Z)$ as instruments, as long as $D$ is binary, $Z$ is discrete, and the relevant assumptions are satisfied. When $Z$ is binary and $Z_{\mathrm{G}} = Z$, we get $\tau(p_{1,x};x) = \tau(x)$ and $\beta_{\mathrm{TSIV}} = \beta_{\mathrm{IV}}$, and we can use equation (\ref{eq:kolesar2}) to write
\begingroup
\allowdisplaybreaks
\begin{eqnarray}
\beta_{\mathrm{IV}} &=& \int \frac{\theta_1(x)}{\int \theta_1(x) \, \mathrm{d} F^X(x)} \, \tau(x) \, \mathrm{d} F^X(x)
\nonumber\\
&=& \frac{\e \left[ \theta_1(X) \cdot \tau(X) \right]}{\e \left[ \theta_1(X) \right]},
\end{eqnarray}
\endgroup
where
\begingroup
\allowdisplaybreaks
\begin{eqnarray}
\theta_1(x) &=& \left( p_{2,x} - p_{1,x} \right) \cdot \pr \left[ P > p_{1,x} \mid X=x \right] \cdot \e \left[ \tilde{P}^L \mid X=x, P > p_{1,x} \right]
\nonumber\\
&=& \big\lvert \, \e \left[ D \mid Z=1, X=x \right] - \e \left[ D \mid Z=0, X=x \right] \big\rvert \cdot \pr \left[ P > p_{1,x} \mid X=x \right]
\nonumber\\
&& \; \cdot \; \e \left[ \tilde{P}^L \mid X=x, P > p_{1,x} \right]
\nonumber\\
&=& \pi(x) \cdot \Big( 1 [ \omega(x) > 0 ] \cdot \pr \left[ Z=1 \mid X=x \right] \cdot \e \left[ \tilde{P}^L \mid X=x, Z=1 \right]
\nonumber\\
&& \; + \; 1 [ \omega(x) < 0 ] \cdot \pr \left[ Z=0 \mid X=x \right] \cdot \e \left[ \tilde{P}^L \mid X=x, Z=0 \right] \Big).
\end{eqnarray}
\endgroup
Next, if $Z_{\mathrm{G}} = Z$, we get $\tilde{P}^L = \mathrm{L} \left[ D \mid Z,X \right] - \mathrm{L} \left[ D \mid X \right]$. If we write $\mathrm{L} \left[ D \mid Z,X \right] = Z \delta + X \zeta$, then $\mathrm{L} \left[ D \mid X \right] = \mathrm{L} \left[ Z \mid X \right] \delta + X \zeta$, which implies that, under Assumption \ref{ass:ps}, $\tilde{P}^L = \left( Z - \mathrm{L} \left[ Z \mid X \right] \right) \delta = \left( Z - \e \left[ Z \mid X \right] \right) \delta$. It follows that $\e \left[ \tilde{P}^L \mid X, Z=1 \right] = \left( 1 - \e \left[ Z \mid X \right] \right) \delta = \pr \left[ Z=0 \mid X \right] \cdot \delta$ and $\e \left[ \tilde{P}^L \mid X, Z=0 \right] = \left( 0 - \e \left[ Z \mid X \right] \right) \delta = - \pr \left[ Z=1 \mid X \right] \cdot \delta$, and further that
\begingroup
\allowdisplaybreaks
\begin{eqnarray}
\theta_1(x) &=& \pi(x) \cdot \Big( 1 [ \omega(x) > 0 ] \cdot \pr \left[ Z=1 \mid X=x \right] \cdot \pr \left[ Z=0 \mid X=x \right] \cdot \delta
\nonumber\\
&& \; - \; 1 [ \omega(x) < 0 ] \cdot \pr \left[ Z=0 \mid X=x \right] \cdot \pr \left[ Z=1 \mid X=x \right] \cdot \delta \Big)
\nonumber\\
&=& \pi(x) \cdot c(x) \cdot \var \left[ Z \mid X=x \right] \cdot \delta,
\end{eqnarray}
\endgroup
which finally implies that
\begingroup
\allowdisplaybreaks
\begin{eqnarray}
\beta_{\mathrm{IV}} &=& \frac{\e \left[ c(X) \cdot \pi(X) \cdot \var \left[ Z \mid X \right] \cdot \delta \cdot \tau(X) \right]}{\e \left[ c(X) \cdot \pi(X) \cdot \var \left[ Z \mid X \right] \cdot \delta \right]}
\nonumber\\
&=& \frac{\e \left[ c(X) \cdot \pi(X) \cdot \var \left[ Z \mid X \right] \cdot \tau(X) \right]}{\e \left[ c(X) \cdot \pi(X) \cdot \var \left[ Z \mid X \right] \right]}.
\end{eqnarray}
\endgroup
This completes the proof.

\bigskip

\paragraph{Proof of Corollary \ref{cor:justid_um}.}

Recall that Assumption \ref{ass:um} is a special case of Assumption \ref{ass:cm} where the existence of compliers but no defiers is postulated at all covariate values and the existence of defiers but no compliers everywhere else (\textit{i.e.}~on an empty set). Thus, it follows from Theorem \ref{the:justid_cm} that, under Assumptions \ref{ass:iv}, \ref{ass:um}, and \ref{ass:ps}, $\beta_{\mathrm{IV}} = \frac{\e \left[ c(X) \cdot \pi(X) \cdot \var \left[ Z \mid X \right] \cdot \tau(X) \right]}{\e \left[ c(X) \cdot \pi(X) \cdot \var \left[ Z \mid X \right] \right]}$ and $c(X)=1$ a.s.

\bigskip

\paragraph{Reordered IV.}

Remark \ref{rem:riv} suggests using $Z_{\mathrm{R}} = 1 [ \omega(X) > 0 ] \cdot Z + 1 [ \omega(X) < 0 ] \cdot \left( 1-Z \right)$ as a new, ``reordered'' instrument in a noninteracted specification. This instrument is binary and takes the value 1 if either $Z=1$ and $\omega(X) > 0$ or $Z=0$ and $\omega(X) < 0$; it also takes the value 0 if either $Z=0$ and $\omega(X) > 0$ or $Z=1$ and $\omega(X) < 0$. It follows that $Z_{\mathrm{R}}$ takes the value 1 for this value of the original instrument that encourages treatment conditional on $X$ and the value 0 otherwise. When we construct the linear IV estimand using $Z_{\mathrm{R}}$ rather than $Z$, we obtain
\begin{equation}
\label{eq:estimand_riv}
\beta_{\mathrm{RIV}} = \left[ \left( \e \left[ Q_{\mathrm{R}}^{\prime} W \right] \right) ^{-1} \e \left[ Q_{\mathrm{R}}^{\prime} Y \right] \right] _1,
\end{equation}
where $Q_{\mathrm{R}} = \left( Z_{\mathrm{R}}, X \right)$ and, as before, $W = \left( D, X \right)$. Formally, we establish the following result.

\begin{corollary}[Reordered IV]
\label{cor:riv}
Suppose that Assumptions \ref{ass:iv} and \ref{ass:cm} hold. Suppose further that $\e \left[ Z_{\mathrm{R}} \mid X \right] = X \alpha_{\mathrm{R}}$\@. Then
\begin{equation*}
\beta_{\mathrm{RIV}} = \frac{\e \left[ \pi(X) \cdot \var \left[ Z \mid X \right] \cdot \tau(X) \right]}{\e \left[ \pi(X) \cdot \var \left[ Z \mid X \right] \right]}.
\end{equation*}
\end{corollary}
\noindent
\emph{Proof.} The assumption that the conditional mean of the instrument is linear in $X$ underlies the proof of Theorem \ref{the:justid_cm}, including equation (\ref{eq:ivdesc}). Under this assumption, we can use equation (\ref{eq:ivdesc}) to write
\begin{equation}
\beta_{\mathrm{RIV}} = \frac{\e \left[ \var \left[ Z_{\mathrm{R}} \mid X \right] \cdot \omega_{\mathrm{R}}(X) \cdot \beta_{\mathrm{R}}(X) \right]}{\e \left[ \var \left[ Z_{\mathrm{R}} \mid X \right] \cdot \omega_{\mathrm{R}}(X) \right]},
\end{equation}
where
\begin{equation}
\omega_{\mathrm{R}}(x) = \e \left[ D \mid Z_{\mathrm{R}}=1, X=x \right] - \e \left[ D \mid Z_{\mathrm{R}}=0, X=x \right]
\end{equation}
and
\begin{equation}
\beta_{\mathrm{R}}(x) = \frac{\phi_{\mathrm{R}}(x)}{\omega_{\mathrm{R}}(x)},
\end{equation}
where
\begin{equation}
\phi_{\mathrm{R}}(x) = \e \left[ Y \mid Z_{\mathrm{R}}=1, X=x \right] - \e \left[ Y \mid Z_{\mathrm{R}}=0, X=x \right].
\end{equation}
Then, it is important to see that $\omega_{\mathrm{R}}(x) = \omega(x)$ and $\phi_{\mathrm{R}}(x) = \phi(x)$ if $\omega(x)>0$, $\omega_{\mathrm{R}}(x) = -\omega(x)$ and $\phi_{\mathrm{R}}(x) = -\phi(x)$ if $\omega(x)<0$, and consequently $\beta_{\mathrm{R}}(x) = \beta(x)$ regardless of the sign of $\omega(x)$. We can also write $\omega_{\mathrm{R}}(x) = c(x) \cdot \omega(x)$, $\phi_{\mathrm{R}}(x) = c(x) \cdot \phi(x)$, and $\var \left[ Z_{\mathrm{R}} \mid X=x \right] = \var \left[ Z \mid X=x \right]$ regardless of the sign of $\omega(x)$. It follows that
\begin{equation}
\beta_{\mathrm{RIV}} = \frac{\e \left[ \var \left[ Z \mid X \right] \cdot c(X) \cdot \omega(X) \cdot \beta(X) \right]}{\e \left[ \var \left[ Z \mid X \right] \cdot c(X) \cdot \omega(X) \right]}.
\end{equation}
To complete this proof, note that, under Assumptions \ref{ass:iv} and \ref{ass:cm}, we know from Lemma \ref{lem:identify} that $\beta(x) = \tau(x)$ and $\omega(x) = c(x) \cdot \pi(x)$. Also, $\left[ c(x) \right] ^2 = 1$ because $c(x) \in \{ -1,1 \}$. Thus, it follows that
\begingroup
\allowdisplaybreaks
\begin{eqnarray}
\beta_{\mathrm{RIV}} &=& \frac{\e \left[ \var \left[ Z \mid X \right] \cdot \left[ c(X) \right] ^2 \cdot \pi(X) \cdot \tau(X) \right]}{\e \left[ \var \left[ Z \mid X \right] \cdot \left[ c(X) \right] ^2 \cdot \pi(X) \right]}
\nonumber\\
&=&
\frac{\e \left[ \var \left[ Z \mid X \right] \cdot \pi(X) \cdot \tau(X) \right]}{\e \left[ \var \left[ Z \mid X \right] \cdot \pi(X) \right]}.
\end{eqnarray}
\endgroup
This completes the proof.

\pagebreak

\section{Simulations}
\label{app:simul}

\vspace{4.175cm}

\setcounter{table}{0}
\renewcommand{\thetable}{\ref{app:simul}.\arabic{table}}

\begin{table}[!h]
\begin{adjustwidth}{-1in}{-1in}
\centering
\begin{threeparttable}
\caption{Simulation Results for $K=250$, ``Strong'' IV, and No Monotonicity Violations\label{tab:sim_250_5}}
\begin{footnotesize}
\begin{tabular}{c>{\centering\arraybackslash}m{0.05cm} >{\centering\arraybackslash}m{1.1425cm} >{\centering\arraybackslash}m{1.1425cm} >{\centering\arraybackslash}m{1.1425cm} >{\centering\arraybackslash}m{0.05cm} >{\centering\arraybackslash}m{1.1425cm} >{\centering\arraybackslash}m{1.1425cm} >{\centering\arraybackslash}m{1.1425cm} >{\centering\arraybackslash}m{0.05cm} >{\centering\arraybackslash}m{1.1425cm} >{\centering\arraybackslash}m{1.1425cm} >{\centering\arraybackslash}m{1.1425cm}}
\hline\hline
          &       & \multicolumn{3}{c}{$N=3{,}000$} &       & \multicolumn{3}{c}{$N=10{,}000$} &       & \multicolumn{3}{c}{$N=50{,}000$} \\
\hline
    A. Estimator Performance &       & Bias  & Median Bias & MSE   &       & Bias  & Median Bias & MSE   &       & Bias  & Median Bias & MSE \\
\hline
    OLS   &       & --0.670 & --0.669 & 3.113 &       & --0.670 & --0.669 & 12.66 &       & --0.669 & --0.669 & 163.1 \\
    IV    &       & 0.005 & 0.000     & 0.121 &       & --0.001 & --0.001 & 0.128 &       & 0.000     & 0.000     & 0.325 \\
    2SLS  &       & --0.371 & --0.372 & 1.000     &       & --0.179 & --0.180 & 1.000     &       & --0.044 & --0.044 & 1.000 \\
    MB2SLS &       & 0.261 & 0.228 & 1.140  &       & 0.025 & 0.023 & 0.218 &       & 0.004 & 0.004 & 0.368 \\
    JIVE  &       & 2.900   & 0.700   & 5.1e+04 &       & 0.393 & 0.383 & 5.073 &       & 0.055 & 0.055 & 1.546 \\
    IJIVE &       & --0.049 & --0.054 & 0.252 &       & --0.006 & --0.010 & 0.182 &       & 0.000     & 0.000     & 0.357 \\
    UJIVE &       & 0.024 & 0.013 & 0.327 &       & 0.001 & --0.002 & 0.186 &       & 0.000     & 0.000     & 0.357 \\
    FEJIV &       & 0.028 & 0.010  & 0.419 &       & 0.001 & 0.000     & 0.188 &       & 0.000     & 0.000     & 0.357 \\
          &       &       &       &       &       &       &       &       &       &       &       &  \\
\hline
    B. Pretest for Weak Identification &       &       &       &       &       &       &       &       &       &       &       &  \\
\hline
    Average $\widetilde{F}$ &       &       & 11.30  &       &       &       & 33.99 &       &       &       & 173.25 &  \\
    $q_{0.05}$ &       &       & 8.04  &       &       &       & 28.50  &       &       &       & 161.83 &  \\
    $q_{0.95}$ &       &       & 14.82 &       &       &       & 39.74 &       &       &       & 184.88 &  \\
\hline
\end{tabular}
\begin{tablenotes}[flushleft]
\item \textit{Notes:} The underlying data-generating process is described in Section \ref{sec:simul}. ``OLS'' is the OLS estimator in the regression of the outcome on the treatment indicator and group indicators. ``IV'' is the IV estimator in the noninteracted specification. The remaining estimators are based on the interacted specification and are described in Section \ref{sec:manyiv}. JIVE, IJIVE, and UJIVE are computed after dropping all groups with fewer than two observations in either $(X,Z)$ combination. FEJIV is computed after dropping all groups with fewer than three observations in either $(X,Z)$ combination. The pretest for weak identification follows \cite{MS2022}; see also the Stata implementation in \cite{Sun2023}. Bias and median bias are reported as the proportion of the target parameter. MSE is normalized by the MSE of 2SLS\@. Results are based on 1,000 replications. Pretest results are based on 250 replications.
\end{tablenotes}
\end{footnotesize}
\end{threeparttable}
\end{adjustwidth}
\end{table}

\begin{table}[!p]
\begin{adjustwidth}{-1in}{-1in}
\centering
\begin{threeparttable}
\caption{Simulation Results for $K=250$, ``Strong'' IV, and Moderate Monotonicity Violations\label{tab:sim_250_6}}
\begin{footnotesize}
\begin{tabular}{c>{\centering\arraybackslash}m{0.05cm} >{\centering\arraybackslash}m{1.1425cm} >{\centering\arraybackslash}m{1.1425cm} >{\centering\arraybackslash}m{1.1425cm} >{\centering\arraybackslash}m{0.05cm} >{\centering\arraybackslash}m{1.1425cm} >{\centering\arraybackslash}m{1.1425cm} >{\centering\arraybackslash}m{1.1425cm} >{\centering\arraybackslash}m{0.05cm} >{\centering\arraybackslash}m{1.1425cm} >{\centering\arraybackslash}m{1.1425cm} >{\centering\arraybackslash}m{1.1425cm}}
\hline\hline
          &       & \multicolumn{3}{c}{$N=3{,}000$} &       & \multicolumn{3}{c}{$N=10{,}000$} &       & \multicolumn{3}{c}{$N=50{,}000$} \\
\hline
    A. Estimator Performance &       & Bias  & Median Bias & MSE   &       & Bias  & Median Bias & MSE   &       & Bias  & Median Bias & MSE \\
\hline
    OLS   &       & --1.051 & --1.050 & 3.204 &       & --1.051 & --1.051 & 11.67 &       & --1.051 & --1.051 & 151.0 \\
    IV    &       & 0.190  & 0.180  & 0.303 &       & 0.172 & 0.173 & 0.499 &       & 0.173 & 0.173 & 4.551 \\
    2SLS  &       & --0.575 & --0.578 & 1.000     &       & --0.297 & --0.299 & 1.000     &       & --0.074 & --0.074 & 1.000 \\
    MB2SLS &       & 0.183 & 0.143 & 0.439 &       & --0.017 & --0.020 & 0.148 &       & --0.004 & --0.005 & 0.305 \\
    JIVE  &       & 68.98 & 1.944 & 1.5e+07 &       & 0.561 & 0.552 & 3.882 &       & 0.081 & 0.080  & 1.273 \\
    IJIVE &       & --0.048 & --0.061 & 0.226 &       & --0.005 & --0.006 & 0.154 &       & 0.000     & --0.001 & 0.306 \\
    UJIVE &       & 0.062 & 0.047 & 0.315 &       & 0.006 & 0.005 & 0.158 &       & 0.001 & 0.000     & 0.307 \\
    FEJIV &       & 0.099 & 0.077 & 0.363 &       & 0.013 & 0.010  & 0.157 &       & 0.001 & 0.000     & 0.307 \\
          &       &       &       &       &       &       &       &       &       &       &       &  \\
\hline
    B. Pretest for Weak Identification &       &       &       &       &       &       &       &       &       &       &       &  \\
\hline
    Average $\widetilde{F}$ &       &       & 11.48 &       &       &       & 33.45 &       &       &       & 162.81 &  \\
    $q_{0.05}$ &       &       & 8.20   &       &       &       & 28.52 &       &       &       & 152.88 &  \\
    $q_{0.95}$ &       &       & 15.08 &       &       &       & 38.69 &       &       &       & 172.65 &  \\
\hline
\end{tabular}
\begin{tablenotes}[flushleft]
\item \textit{Notes:} The underlying data-generating process is described in Section \ref{sec:simul}. ``OLS'' is the OLS estimator in the regression of the outcome on the treatment indicator and group indicators. ``IV'' is the IV estimator in the noninteracted specification. The remaining estimators are based on the interacted specification and are described in Section \ref{sec:manyiv}. JIVE, IJIVE, and UJIVE are computed after dropping all groups with fewer than two observations in either $(X,Z)$ combination. FEJIV is computed after dropping all groups with fewer than three observations in either $(X,Z)$ combination. The pretest for weak identification follows \cite{MS2022}; see also the Stata implementation in \cite{Sun2023}. Bias and median bias are reported as the proportion of the target parameter. MSE is normalized by the MSE of 2SLS\@. Results are based on 1,000 replications. Pretest results are based on 250 replications.
\item \phantom{extra space here}
\end{tablenotes}
\end{footnotesize}
\end{threeparttable}
\end{adjustwidth}
\end{table}

\begin{table}[!p]
\begin{adjustwidth}{-1in}{-1in}
\centering
\begin{threeparttable}
\caption{Simulation Results for $K=250$, ``Strong'' IV, and Large Monotonicity Violations\label{tab:sim_250_7}}
\begin{footnotesize}
\begin{tabular}{c>{\centering\arraybackslash}m{0.05cm} >{\centering\arraybackslash}m{1.1425cm} >{\centering\arraybackslash}m{1.1425cm} >{\centering\arraybackslash}m{1.1425cm} >{\centering\arraybackslash}m{0.05cm} >{\centering\arraybackslash}m{1.1425cm} >{\centering\arraybackslash}m{1.1425cm} >{\centering\arraybackslash}m{1.1425cm} >{\centering\arraybackslash}m{0.05cm} >{\centering\arraybackslash}m{1.1425cm} >{\centering\arraybackslash}m{1.1425cm} >{\centering\arraybackslash}m{1.1425cm}}
\hline\hline
          &       & \multicolumn{3}{c}{$N=3{,}000$} &       & \multicolumn{3}{c}{$N=10{,}000$} &       & \multicolumn{3}{c}{$N=50{,}000$} \\
\hline
    A. Estimator Performance &       & Bias  & Median Bias & MSE   &       & Bias  & Median Bias & MSE   &       & Bias  & Median Bias & MSE \\
\hline
    OLS   &       & --1.058 & --1.059 & 3.132 &       & --1.060 & --1.060 & 11.29 &       & --1.059 & --1.059 & 147.4 \\
    IV    &       & 0.729 & 0.546 & 4.686 &       & 0.564 & 0.551 & 4.983 &       & 0.545 & 0.541 & 43.22 \\
    2SLS  &       & --0.586 & --0.588 & 1.000     &       & --0.305 & --0.307 & 1.000     &       & --0.077 & --0.078 & 1.000 \\
    MB2SLS &       & 0.038 & 0.024 & 0.251 &       & --0.047 & --0.050 & 0.151 &       & --0.011 & --0.013 & 0.271 \\
    JIVE  &       & --11.89 & 2.337 & 1.4e+06 &       & 0.527 & 0.515 & 3.309 &       & 0.077 & 0.075 & 1.091 \\
    IJIVE &       & --0.058 & --0.073 & 0.234 &       & --0.005 & --0.007 & 0.148 &       & 0.000     & --0.002 & 0.262 \\
    UJIVE &       & 0.056 & 0.030  & 0.327 &       & 0.007 & 0.006 & 0.153 &       & 0.000     & --0.001 & 0.262 \\
    FEJIV &       & 0.099 & 0.074 & 0.357 &       & 0.014 & 0.012 & 0.153 &       & 0.000     & --0.001 & 0.262 \\
          &       &       &       &       &       &       &       &       &       &       &       &  \\
\hline
    B. Pretest for Weak Identification &       &       &       &       &       &       &       &       &       &       &       &  \\
\hline
    Average $\widetilde{F}$ &       &       & 11.61 &       &       &       & 32.16 &       &       &       & 155.20 &  \\
    $q_{0.05}$ &       &       & 8.07  &       &       &       & 27.72 &       &       &       & 146.33 &  \\
    $q_{0.95}$ &       &       & 15.66 &       &       &       & 36.54 &       &       &       & 165.50 &  \\
\hline
\end{tabular}
\begin{tablenotes}[flushleft]
\item \textit{Notes:} The underlying data-generating process is described in Section \ref{sec:simul}. ``OLS'' is the OLS estimator in the regression of the outcome on the treatment indicator and group indicators. ``IV'' is the IV estimator in the noninteracted specification. The remaining estimators are based on the interacted specification and are described in Section \ref{sec:manyiv}. JIVE, IJIVE, and UJIVE are computed after dropping all groups with fewer than two observations in either $(X,Z)$ combination. FEJIV is computed after dropping all groups with fewer than three observations in either $(X,Z)$ combination. The pretest for weak identification follows \cite{MS2022}; see also the Stata implementation in \cite{Sun2023}. Bias and median bias are reported as the proportion of the target parameter. MSE is normalized by the MSE of 2SLS\@. Results are based on 1,000 replications. Pretest results are based on 250 replications.
\item \phantom{extra space here}
\end{tablenotes}
\end{footnotesize}
\end{threeparttable}
\end{adjustwidth}
\end{table}

\begin{table}[!p]
\begin{adjustwidth}{-1in}{-1in}
\centering
\begin{threeparttable}
\caption{Simulation Results for $K=250$, ``Strong'' IV, and Monotonicity Violations with Weak Cells\label{tab:sim_250_8}}
\begin{footnotesize}
\begin{tabular}{c>{\centering\arraybackslash}m{0.05cm} >{\centering\arraybackslash}m{1.1425cm} >{\centering\arraybackslash}m{1.1425cm} >{\centering\arraybackslash}m{1.1425cm} >{\centering\arraybackslash}m{0.05cm} >{\centering\arraybackslash}m{1.1425cm} >{\centering\arraybackslash}m{1.1425cm} >{\centering\arraybackslash}m{1.1425cm} >{\centering\arraybackslash}m{0.05cm} >{\centering\arraybackslash}m{1.1425cm} >{\centering\arraybackslash}m{1.1425cm} >{\centering\arraybackslash}m{1.1425cm}}
\hline\hline
          &       & \multicolumn{3}{c}{$N=3{,}000$} &       & \multicolumn{3}{c}{$N=10{,}000$} &       & \multicolumn{3}{c}{$N=50{,}000$} \\
\hline
    A. Estimator Performance &       & Bias  & Median Bias & MSE   &       & Bias  & Median Bias & MSE   &       & Bias  & Median Bias & MSE \\
\hline
    OLS   &       & --1.127 & --1.128 & 2.354 &       & --1.125 & --1.126 & 6.550  &       & --1.125 & --1.125 & 72.16 \\
    IV    &       & 0.811 & 0.629 & 24.01 &       & 0.678 & 0.600   & 4.107 &       & 0.641 & 0.633 & 26.21 \\
    2SLS  &       & --0.721 & --0.726 & 1.000     &       & --0.427 & --0.429 & 1.000     &       & --0.120 & --0.119 & 1.000 \\
    MB2SLS &       & 0.146 & 0.079 & 0.600   &       & --0.072 & --0.076 & 0.167 &       & --0.018 & --0.018 & 0.241 \\
    JIVE  &       & --3.500 & --3.108 & 257.9 &       & 1.185 & 1.076 & 9.265 &       & 0.128 & 0.128 & 1.235 \\
    IJIVE &       & --0.058 & --0.108 & 0.429 &       & --0.011 & --0.020 & 0.174 &       & --0.001 & --0.002 & 0.231 \\
    UJIVE &       & 0.164 & 0.058 & 0.996 &       & 0.008 & --0.002 & 0.183 &       & --0.001 & --0.001 & 0.232 \\
    FEJIV &       & 0.220  & 0.116 & 1.757 &       & 0.019 & 0.008 & 0.180  &       & --0.001 & --0.001 & 0.232 \\
          &       &       &       &       &       &       &       &       &       &       &       &  \\
\hline
    B. Pretest for Weak Identification &       &       &       &       &       &       &       &       &       &       &       &  \\
\hline
    Average $\widetilde{F}$ &       &       & 7.64  &       &       &       & 20.47 &       &       &       & 99.71 &  \\
    $q_{0.05}$ &       &       & 5.09  &       &       &       & 16.36 &       &       &       & 92.88 &  \\
    $q_{0.95}$ &       &       & 10.82 &       &       &       & 24.61 &       &       &       & 107.47 &  \\
\hline
\end{tabular}
\begin{tablenotes}[flushleft]
\item \textit{Notes:} The underlying data-generating process is described in Section \ref{sec:simul}. ``OLS'' is the OLS estimator in the regression of the outcome on the treatment indicator and group indicators. ``IV'' is the IV estimator in the noninteracted specification. The remaining estimators are based on the interacted specification and are described in Section \ref{sec:manyiv}. JIVE, IJIVE, and UJIVE are computed after dropping all groups with fewer than two observations in either $(X,Z)$ combination. FEJIV is computed after dropping all groups with fewer than three observations in either $(X,Z)$ combination. The pretest for weak identification follows \cite{MS2022}; see also the Stata implementation in \cite{Sun2023}. Bias and median bias are reported as the proportion of the target parameter. MSE is normalized by the MSE of 2SLS\@. Results are based on 1,000 replications. Pretest results are based on 250 replications.
\item \phantom{extra space here}
\end{tablenotes}
\end{footnotesize}
\end{threeparttable}
\end{adjustwidth}
\end{table}

\begin{table}[!p]
\begin{adjustwidth}{-1in}{-1in}
\centering
\begin{threeparttable}
\caption{Simulation Results for $K=20$, ``Weak'' IV, and No Monotonicity Violations\label{tab:sim_20_1}}
\begin{footnotesize}
\begin{tabular}{c>{\centering\arraybackslash}m{0.05cm} >{\centering\arraybackslash}m{1.45cm} >{\centering\arraybackslash}m{1.45cm} >{\centering\arraybackslash}m{1.45cm} >{\centering\arraybackslash}m{0.05cm} >{\centering\arraybackslash}m{1.45cm} >{\centering\arraybackslash}m{1.45cm} >{\centering\arraybackslash}m{1.45cm}}
\hline\hline
          &       & \multicolumn{3}{c}{$N=3{,}000$} &       & \multicolumn{3}{c}{$N=10{,}000$} \\
\hline
    A. Estimator Performance &       & Bias  & Median Bias & MSE   &       & Bias  & Median Bias & MSE \\
\hline
    OLS   &       & --0.737 & --0.736 & 1.875 &       & --0.737 & --0.736 & 5.975 \\
    IV    &       & 0.070  & --0.004 & 1.368 &       & 0.024 & 0.010  & 0.902 \\
    2SLS  &       & --0.426 & --0.419 & 1.000     &       & --0.199 & --0.204 & 1.000 \\
    MB2SLS &       & 0.696 & --0.009 & 193.4 &       & 0.053 & 0.021 & 1.569 \\
    JIVE  &       & --1.586 & --1.084 & 1.2e+03  &       & 1.092 & 0.559 & 3.0e+03 \\
    IJIVE &       & 0.384 & --0.055 & 440.5 &       & 0.058 & 0.023 & 1.671 \\
    UJIVE &       & --1.492 & --0.039 & 8.0e+04 &       & 0.065 & 0.026 & 1.755 \\
    FEJIV &       & --1.280 & --0.012 & 2.1e+03  &       & 0.063 & 0.027 & 1.700 \\
          &       &       &       &       &       &       &       &  \\
\hline
    B. Pretest for Weak Identification &       &       &       &       &       &       &       &  \\
\hline
    Average $\widetilde{F}$ &       &       & 2.80   &       &       &       & 7.84  &  \\
    $q_{0.05}$ &       &       & 0.28  &       &       &       & 4.14  &  \\
    $q_{0.95}$ &       &       & 6.15  &       &       &       & 12.49 &  \\
\hline
\end{tabular}
\begin{tablenotes}[flushleft]
\item \textit{Notes:} The underlying data-generating process is described in Section \ref{sec:simul}. ``OLS'' is the OLS estimator in the regression of the outcome on the treatment indicator and group indicators. ``IV'' is the IV estimator in the noninteracted specification. The remaining estimators are based on the interacted specification and are described in Section \ref{sec:manyiv}. JIVE, IJIVE, and UJIVE are computed after dropping all groups with fewer than two observations in either $(X,Z)$ combination. FEJIV is computed after dropping all groups with fewer than three observations in either $(X,Z)$ combination. The pretest for weak identification follows \cite{MS2022}; see also the Stata implementation in \cite{Sun2023}. Bias and median bias are reported as the proportion of the target parameter. MSE is normalized by the MSE of 2SLS\@. Results are based on 1,000 replications. Pretest results are based on 250 replications.
\item \phantom{extra space here}
\end{tablenotes}
\end{footnotesize}
\end{threeparttable}
\end{adjustwidth}
\end{table}

\begin{table}[!p]
\begin{adjustwidth}{-1in}{-1in}
\centering
\begin{threeparttable}
\caption{Simulation Results for $K=20$, ``Weak'' IV, and Moderate Monotonicity Violations\label{tab:sim_20_2}}
\begin{footnotesize}
\begin{tabular}{c>{\centering\arraybackslash}m{0.05cm} >{\centering\arraybackslash}m{1.45cm} >{\centering\arraybackslash}m{1.45cm} >{\centering\arraybackslash}m{1.45cm} >{\centering\arraybackslash}m{0.05cm} >{\centering\arraybackslash}m{1.45cm} >{\centering\arraybackslash}m{1.45cm} >{\centering\arraybackslash}m{1.45cm}}
\hline\hline
          &       & \multicolumn{3}{c}{$N=3{,}000$} &       & \multicolumn{3}{c}{$N=10{,}000$} \\
\hline
    A. Estimator Performance &       & Bias  & Median Bias & MSE   &       & Bias  & Median Bias & MSE \\
\hline
    OLS   &       & --1.134 & --1.132 & 1.949 &       & --1.133 & --1.132 & 6.075 \\
    IV    &       & 2.716 & 0.186 & 8.0e+03  &       & 0.277 & 0.199 & 2.317 \\
    2SLS  &       & --0.663 & --0.652 & 1.000     &       & --0.312 & --0.318 & 1.000 \\
    MB2SLS &       & 0.295 & --0.095 & 386.4 &       & 0.051 & --0.011 & 1.540 \\
    JIVE  &       & --2.849 & --1.689 & 1.6e+03  &       & --3.561 & 0.768 & 6.6e+04 \\
    IJIVE &       & 1.833 & --0.072 & 3.4e+03  &       & 0.092 & 0.031 & 1.801 \\
    UJIVE &       & 2.511 & --0.061 & 1.2e+04 &       & 0.102 & 0.035 & 1.899 \\
    FEJIV &       & 1.033 & --0.027 & 1.1e+03  &       & 0.099 & 0.033 & 1.769 \\
          &       &       &       &       &       &       &       &  \\
\hline
    B. Pretest for Weak Identification &       &       &       &       &       &       &       &  \\
\hline
    Average $\widetilde{F}$ &       &       & 2.62  &       &       &       & 7.75  &  \\
    $q_{0.05}$ &       &       & 0.24  &       &       &       & 3.90   &  \\
    $q_{0.95}$ &       &       & 5.50   &       &       &       & 12.03 &  \\
\hline
\end{tabular}
\begin{tablenotes}[flushleft]
\item \textit{Notes:} The underlying data-generating process is described in Section \ref{sec:simul}. ``OLS'' is the OLS estimator in the regression of the outcome on the treatment indicator and group indicators. ``IV'' is the IV estimator in the noninteracted specification. The remaining estimators are based on the interacted specification and are described in Section \ref{sec:manyiv}. JIVE, IJIVE, and UJIVE are computed after dropping all groups with fewer than two observations in either $(X,Z)$ combination. FEJIV is computed after dropping all groups with fewer than three observations in either $(X,Z)$ combination. The pretest for weak identification follows \cite{MS2022}; see also the Stata implementation in \cite{Sun2023}. Bias and median bias are reported as the proportion of the target parameter. MSE is normalized by the MSE of 2SLS\@. Results are based on 1,000 replications. Pretest results are based on 250 replications.
\item \phantom{extra space here}
\end{tablenotes}
\end{footnotesize}
\end{threeparttable}
\end{adjustwidth}
\end{table}

\begin{table}[!p]
\begin{adjustwidth}{-1in}{-1in}
\centering
\begin{threeparttable}
\caption{Simulation Results for $K=20$, ``Weak'' IV, and Large Monotonicity Violations\label{tab:sim_20_3}}
\begin{footnotesize}
\begin{tabular}{c>{\centering\arraybackslash}m{0.05cm} >{\centering\arraybackslash}m{1.45cm} >{\centering\arraybackslash}m{1.45cm} >{\centering\arraybackslash}m{1.45cm} >{\centering\arraybackslash}m{0.05cm} >{\centering\arraybackslash}m{1.45cm} >{\centering\arraybackslash}m{1.45cm} >{\centering\arraybackslash}m{1.45cm}}
\hline\hline
          &       & \multicolumn{3}{c}{$N=3{,}000$} &       & \multicolumn{3}{c}{$N=10{,}000$} \\
\hline
    A. Estimator Performance &       & Bias  & Median Bias & MSE   &       & Bias  & Median Bias & MSE \\
\hline
    OLS   &       & --1.138 & --1.137 & 1.951 &       & --1.138 & --1.137 & 6.109 \\
    IV    &       & 0.602 & 0.074 & 506.6 &       & 0.802 & 0.331 & 67.86 \\
    2SLS  &       & --0.663 & --0.649 & 1.000     &       & --0.324 & --0.332 & 1.000 \\
    MB2SLS &       & 1.209 & --0.122 & 1.1e+03  &       & 0.025 & --0.031 & 1.331 \\
    JIVE  &       & --0.965 & --1.781 & 8.5e+03  &       & 0.440  & 0.752 & 3.1e+03 \\
    IJIVE &       & --0.203 & --0.099 & 183.9 &       & 0.084 & 0.003 & 1.653 \\
    UJIVE &       & 0.093 & --0.090 & 264.9 &       & 0.094 & 0.011 & 1.736 \\
    FEJIV &       & 0.369 & --0.003 & 239.8 &       & 0.090  & 0.014 & 1.645 \\
          &       &       &       &       &       &       &       &  \\
\hline
    B. Pretest for Weak Identification &       &       &       &       &       &       &       &  \\
\hline
    Average $\widetilde{F}$ &       &       & 2.54  &       &       &       & 7.48  &  \\
    $q_{0.05}$ &       &       & 0.02  &       &       &       & 3.80   &  \\
    $q_{0.95}$ &       &       & 5.97  &       &       &       & 11.62 &  \\
\hline
\end{tabular}
\begin{tablenotes}[flushleft]
\item \textit{Notes:} The underlying data-generating process is described in Section \ref{sec:simul}. ``OLS'' is the OLS estimator in the regression of the outcome on the treatment indicator and group indicators. ``IV'' is the IV estimator in the noninteracted specification. The remaining estimators are based on the interacted specification and are described in Section \ref{sec:manyiv}. JIVE, IJIVE, and UJIVE are computed after dropping all groups with fewer than two observations in either $(X,Z)$ combination. FEJIV is computed after dropping all groups with fewer than three observations in either $(X,Z)$ combination. The pretest for weak identification follows \cite{MS2022}; see also the Stata implementation in \cite{Sun2023}. Bias and median bias are reported as the proportion of the target parameter. MSE is normalized by the MSE of 2SLS\@. Results are based on 1,000 replications. Pretest results are based on 250 replications.
\item \phantom{extra space here}
\end{tablenotes}
\end{footnotesize}
\end{threeparttable}
\end{adjustwidth}
\end{table}

\begin{table}[!p]
\begin{adjustwidth}{-1in}{-1in}
\centering
\begin{threeparttable}
\caption{Simulation Results for $K=20$, ``Weak'' IV, and Monotonicity Violations with Weak Cells\label{tab:sim_20_4}}
\begin{footnotesize}
\begin{tabular}{c>{\centering\arraybackslash}m{0.05cm} >{\centering\arraybackslash}m{1.45cm} >{\centering\arraybackslash}m{1.45cm} >{\centering\arraybackslash}m{1.45cm} >{\centering\arraybackslash}m{0.05cm} >{\centering\arraybackslash}m{1.45cm} >{\centering\arraybackslash}m{1.45cm} >{\centering\arraybackslash}m{1.45cm}}
\hline\hline
          &       & \multicolumn{3}{c}{$N=3{,}000$} &       & \multicolumn{3}{c}{$N=10{,}000$} \\
\hline
    A. Estimator Performance &       & Bias  & Median Bias & MSE   &       & Bias  & Median Bias & MSE \\
\hline
    OLS   &       & --1.152 & --1.151 & 1.633 &       & --1.149 & --1.149 & 4.847 \\
    IV    &       & 1.409 & 0.139 & 1.3e+03  &       & --0.872 & 0.323 & 7.7e+03 \\
    2SLS  &       & --0.740 & --0.749 & 1.000     &       & --0.371 & --0.371 & 1.000 \\
    MB2SLS &       & --0.276 & --0.290 & 104.8 &       & 0.059 & --0.005 & 1.658 \\
    JIVE  &       & --0.889 & --1.696 & 1.8e+03  &       & 0.803 & 0.969 & 7.9e+03 \\
    IJIVE &       & --0.672 & --0.202 & 892.2 &       & 0.131 & 0.053 & 2.075 \\
    UJIVE &       & --0.751 & --0.212 & 377.9 &       & 0.148 & 0.067 & 2.243 \\
    FEJIV &       & 0.265 & --0.122 & 351.7 &       & 0.138 & 0.064 & 2.130 \\
          &       &       &       &       &       &       &       &  \\
\hline
    B. Pretest for Weak Identification &       &       &       &       &       &       &       &  \\
\hline
    Average $\widetilde{F}$ &       &       & 2.18  &       &       &       & 6.09  &  \\
    $q_{0.05}$ &       &       & --0.02 &       &       &       & 2.32  &  \\
    $q_{0.95}$ &       &       & 4.66  &       &       &       & 9.85  &  \\
\hline
\end{tabular}
\begin{tablenotes}[flushleft]
\item \textit{Notes:} The underlying data-generating process is described in Section \ref{sec:simul}. ``OLS'' is the OLS estimator in the regression of the outcome on the treatment indicator and group indicators. ``IV'' is the IV estimator in the noninteracted specification. The remaining estimators are based on the interacted specification and are described in Section \ref{sec:manyiv}. JIVE, IJIVE, and UJIVE are computed after dropping all groups with fewer than two observations in either $(X,Z)$ combination. FEJIV is computed after dropping all groups with fewer than three observations in either $(X,Z)$ combination. The pretest for weak identification follows \cite{MS2022}; see also the Stata implementation in \cite{Sun2023}. Bias and median bias are reported as the proportion of the target parameter. MSE is normalized by the MSE of 2SLS\@. Results are based on 1,000 replications. Pretest results are based on 250 replications.
\item \phantom{extra space here}
\end{tablenotes}
\end{footnotesize}
\end{threeparttable}
\end{adjustwidth}
\end{table}

\begin{table}[!p]
\begin{adjustwidth}{-1in}{-1in}
\centering
\begin{threeparttable}
\caption{Simulation Results for $K=20$, ``Strong'' IV, and No Monotonicity Violations\label{tab:sim_20_5}}
\begin{footnotesize}
\begin{tabular}{c>{\centering\arraybackslash}m{0.05cm} >{\centering\arraybackslash}m{1.45cm} >{\centering\arraybackslash}m{1.45cm} >{\centering\arraybackslash}m{1.45cm} >{\centering\arraybackslash}m{0.05cm} >{\centering\arraybackslash}m{1.45cm} >{\centering\arraybackslash}m{1.45cm} >{\centering\arraybackslash}m{1.45cm}}
\hline\hline
          &       & \multicolumn{3}{c}{$N=3{,}000$} &       & \multicolumn{3}{c}{$N=10{,}000$} \\
\hline
    A. Estimator Performance &       & Bias  & Median Bias & MSE   &       & Bias  & Median Bias & MSE \\
\hline
    OLS   &       & --0.669 & --0.669 & 29.70  &       & --0.669 & --0.668 & 105.8 \\
    IV    &       & 0.003 & --0.001 & 0.949 &       & 0.002 & 0.002 & 0.983 \\
    2SLS  &       & --0.049 & --0.052 & 1.000     &       & --0.014 & --0.015 & 1.000 \\
    MB2SLS &       & 0.009 & 0.005 & 1.053 &       & 0.004 & 0.003 & 1.021 \\
    JIVE  &       & 0.072 & 0.063 & 1.679 &       & 0.023 & 0.022 & 1.226 \\
    IJIVE &       & 0.004 & 0.001 & 1.027 &       & 0.003 & 0.002 & 1.020 \\
    UJIVE &       & 0.007 & 0.002 & 1.047 &       & 0.004 & 0.002 & 1.023 \\
    FEJIV &       & 0.007 & 0.001 & 1.033 &       & 0.004 & 0.002 & 1.020 \\
          &       &       &       &       &       &       &       &  \\
\hline
    B. Pretest for Weak Identification &       &       &       &       &       &       &       &  \\
\hline
    Average $\widetilde{F}$ &       &       & 42.75 &       &       &       & 131.58 &  \\
    $q_{0.05}$ &       &       & 31.36 &       &       &       & 113.22 &  \\
    $q_{0.95}$ &       &       & 55.10  &       &       &       & 154.91 &  \\
\hline
\end{tabular}
\begin{tablenotes}[flushleft]
\item \textit{Notes:} The underlying data-generating process is described in Section \ref{sec:simul}. ``OLS'' is the OLS estimator in the regression of the outcome on the treatment indicator and group indicators. ``IV'' is the IV estimator in the noninteracted specification. The remaining estimators are based on the interacted specification and are described in Section \ref{sec:manyiv}. JIVE, IJIVE, and UJIVE are computed after dropping all groups with fewer than two observations in either $(X,Z)$ combination. FEJIV is computed after dropping all groups with fewer than three observations in either $(X,Z)$ combination. The pretest for weak identification follows \cite{MS2022}; see also the Stata implementation in \cite{Sun2023}. Bias and median bias are reported as the proportion of the target parameter. MSE is normalized by the MSE of 2SLS\@. Results are based on 1,000 replications. Pretest results are based on 250 replications.
\item \phantom{extra space here}
\end{tablenotes}
\end{footnotesize}
\end{threeparttable}
\end{adjustwidth}
\end{table}

\begin{table}[!p]
\begin{adjustwidth}{-1in}{-1in}
\centering
\begin{threeparttable}
\caption{Simulation Results for $K=20$, ``Strong'' IV, and Moderate Monotonicity Violations\label{tab:sim_20_6}}
\begin{footnotesize}
\begin{tabular}{c>{\centering\arraybackslash}m{0.05cm} >{\centering\arraybackslash}m{1.45cm} >{\centering\arraybackslash}m{1.45cm} >{\centering\arraybackslash}m{1.45cm} >{\centering\arraybackslash}m{0.05cm} >{\centering\arraybackslash}m{1.45cm} >{\centering\arraybackslash}m{1.45cm} >{\centering\arraybackslash}m{1.45cm}}
\hline\hline
          &       & \multicolumn{3}{c}{$N=3{,}000$} &       & \multicolumn{3}{c}{$N=10{,}000$} \\
\hline
    A. Estimator Performance &       & Bias  & Median Bias & MSE   &       & Bias  & Median Bias & MSE \\
\hline
    OLS   &       & --1.027 & --1.027 & 32.72 &       & --1.025 & --1.025 & 114.2 \\
    IV    &       & 0.213 & 0.196 & 3.545 &       & 0.206 & 0.200   & 6.702 \\
    2SLS  &       & --0.084 & --0.089 & 1.000     &       & --0.021 & --0.022 & 1.000 \\
    MB2SLS &       & --0.003 & --0.008 & 0.958 &       & 0.004 & 0.003 & 1.012 \\
    JIVE  &       & 0.099 & 0.094 & 1.562 &       & 0.036 & 0.033 & 1.237 \\
    IJIVE &       & 0.002 & 0.002 & 0.980  &       & 0.007 & 0.005 & 1.025 \\
    UJIVE &       & 0.006 & 0.005 & 0.999 &       & 0.007 & 0.006 & 1.028 \\
    FEJIV &       & 0.006 & 0.003 & 0.986 &       & 0.007 & 0.007 & 1.025 \\
          &       &       &       &       &       &       &       &  \\
\hline
    B. Pretest for Weak Identification &       &       &       &       &       &       &       &  \\
\hline
    Average $\widetilde{F}$ &       &       & 40.58 &       &       &       & 124.89 &  \\
    $q_{0.05}$ &       &       & 31.14 &       &       &       & 107.54 &  \\
    $q_{0.95}$ &       &       & 51.84 &       &       &       & 140.35 &  \\
\hline
\end{tabular}
\begin{tablenotes}[flushleft]
\item \textit{Notes:} The underlying data-generating process is described in Section \ref{sec:simul}. ``OLS'' is the OLS estimator in the regression of the outcome on the treatment indicator and group indicators. ``IV'' is the IV estimator in the noninteracted specification. The remaining estimators are based on the interacted specification and are described in Section \ref{sec:manyiv}. JIVE, IJIVE, and UJIVE are computed after dropping all groups with fewer than two observations in either $(X,Z)$ combination. FEJIV is computed after dropping all groups with fewer than three observations in either $(X,Z)$ combination. The pretest for weak identification follows \cite{MS2022}; see also the Stata implementation in \cite{Sun2023}. Bias and median bias are reported as the proportion of the target parameter. MSE is normalized by the MSE of 2SLS\@. Results are based on 1,000 replications. Pretest results are based on 250 replications.
\item \phantom{extra space here}
\end{tablenotes}
\end{footnotesize}
\end{threeparttable}
\end{adjustwidth}
\end{table}

\begin{table}[!p]
\begin{adjustwidth}{-1in}{-1in}
\centering
\begin{threeparttable}
\caption{Simulation Results for $K=20$, ``Strong'' IV, and Large Monotonicity Violations\label{tab:sim_20_7}}
\begin{footnotesize}
\begin{tabular}{c>{\centering\arraybackslash}m{0.05cm} >{\centering\arraybackslash}m{1.45cm} >{\centering\arraybackslash}m{1.45cm} >{\centering\arraybackslash}m{1.45cm} >{\centering\arraybackslash}m{0.05cm} >{\centering\arraybackslash}m{1.45cm} >{\centering\arraybackslash}m{1.45cm} >{\centering\arraybackslash}m{1.45cm}}
\hline\hline
          &       & \multicolumn{3}{c}{$N=3{,}000$} &       & \multicolumn{3}{c}{$N=10{,}000$} \\
\hline
    A. Estimator Performance &       & Bias  & Median Bias & MSE   &       & Bias  & Median Bias & MSE \\
\hline
    OLS   &       & --1.032 & --1.033 & 31.52 &       & --1.032 & --1.030 & 117.0 \\
    IV    &       & 0.397 & 0.333 & 13.87 &       & 0.357 & 0.339 & 22.19 \\
    2SLS  &       & --0.087 & --0.090 & 1.000     &       & --0.026 & --0.026 & 1.000 \\
    MB2SLS &       & --0.008 & --0.012 & 0.950  &       & --0.001 & --0.002 & 0.987 \\
    JIVE  &       & 0.096 & 0.086 & 1.532 &       & 0.031 & 0.030  & 1.182 \\
    IJIVE &       & 0.002 & 0.001 & 0.986 &       & 0.004 & 0.003 & 1.002 \\
    UJIVE &       & 0.006 & 0.004 & 1.007 &       & 0.004 & 0.003 & 1.005 \\
    FEJIV &       & 0.007 & 0.002 & 0.989 &       & 0.004 & 0.003 & 1.002 \\
          &       &       &       &       &       &       &       &  \\
\hline
    B. Pretest for Weak Identification &       &       &       &       &       &       &       &  \\
\hline
    Average $\widetilde{F}$ &       &       & 39.01 &       &       &       & 119.92 &  \\
    $q_{0.05}$ &       &       & 28.85 &       &       &       & 104.08 &  \\
    $q_{0.95}$ &       &       & 49.38 &       &       &       & 136.68 &  \\
\hline
\end{tabular}
\begin{tablenotes}[flushleft]
\item \textit{Notes:} The underlying data-generating process is described in Section \ref{sec:simul}. ``OLS'' is the OLS estimator in the regression of the outcome on the treatment indicator and group indicators. ``IV'' is the IV estimator in the noninteracted specification. The remaining estimators are based on the interacted specification and are described in Section \ref{sec:manyiv}. JIVE, IJIVE, and UJIVE are computed after dropping all groups with fewer than two observations in either $(X,Z)$ combination. FEJIV is computed after dropping all groups with fewer than three observations in either $(X,Z)$ combination. The pretest for weak identification follows \cite{MS2022}; see also the Stata implementation in \cite{Sun2023}. Bias and median bias are reported as the proportion of the target parameter. MSE is normalized by the MSE of 2SLS\@. Results are based on 1,000 replications. Pretest results are based on 250 replications.
\item \phantom{extra space here}
\end{tablenotes}
\end{footnotesize}
\end{threeparttable}
\end{adjustwidth}
\end{table}

\begin{table}[!p]
\begin{adjustwidth}{-1in}{-1in}
\centering
\begin{threeparttable}
\caption{Simulation Results for $K=20$, ``Strong'' IV, and Monotonicity Violations with Weak Cells\label{tab:sim_20_8}}
\begin{footnotesize}
\begin{tabular}{c>{\centering\arraybackslash}m{0.05cm} >{\centering\arraybackslash}m{1.45cm} >{\centering\arraybackslash}m{1.45cm} >{\centering\arraybackslash}m{1.45cm} >{\centering\arraybackslash}m{0.05cm} >{\centering\arraybackslash}m{1.45cm} >{\centering\arraybackslash}m{1.45cm} >{\centering\arraybackslash}m{1.45cm}}
\hline\hline
          &       & \multicolumn{3}{c}{$N=3{,}000$} &       & \multicolumn{3}{c}{$N=10{,}000$} \\
\hline
    A. Estimator Performance &       & Bias  & Median Bias & MSE   &       & Bias  & Median Bias & MSE \\
\hline
    OLS   &       & --1.062 & --1.064 & 25.17 &       & --1.061 & --1.062 & 91.49 \\
    IV    &       & 0.374 & 0.325 & 7.865 &       & 0.340  & 0.333 & 14.39 \\
    2SLS  &       & --0.101 & --0.101 & 1.000     &       & --0.032 & --0.030 & 1.000 \\
    MB2SLS &       & --0.004 & --0.011 & 0.987 &       & --0.001 & 0.001 & 0.990 \\
    JIVE  &       & 0.123 & 0.101 & 1.732 &       & 0.038 & 0.042 & 1.206 \\
    IJIVE &       & 0.006 & --0.002 & 1.022 &       & 0.004 & 0.007 & 1.005 \\
    UJIVE &       & 0.011 & 0.002 & 1.052 &       & 0.004 & 0.008 & 1.008 \\
    FEJIV &       & 0.012 & 0.004 & 1.008 &       & 0.004 & 0.008 & 1.007 \\
          &       &       &       &       &       &       &       &  \\
\hline
    B. Pretest for Weak Identification &       &       &       &       &       &       &       &  \\
\hline
    Average $\widetilde{F}$ &       &       & 33.37 &       &       &       & 101.53 &  \\
    $q_{0.05}$ &       &       & 24.50  &       &       &       & 85.67 &  \\
    $q_{0.95}$ &       &       & 42.08 &       &       &       & 117.25 &  \\
\hline
\end{tabular}
\begin{tablenotes}[flushleft]
\item \textit{Notes:} The underlying data-generating process is described in Section \ref{sec:simul}. ``OLS'' is the OLS estimator in the regression of the outcome on the treatment indicator and group indicators. ``IV'' is the IV estimator in the noninteracted specification. The remaining estimators are based on the interacted specification and are described in Section \ref{sec:manyiv}. JIVE, IJIVE, and UJIVE are computed after dropping all groups with fewer than two observations in either $(X,Z)$ combination. FEJIV is computed after dropping all groups with fewer than three observations in either $(X,Z)$ combination. The pretest for weak identification follows \cite{MS2022}; see also the Stata implementation in \cite{Sun2023}. Bias and median bias are reported as the proportion of the target parameter. MSE is normalized by the MSE of 2SLS\@. Results are based on 1,000 replications. Pretest results are based on 250 replications.
\item \phantom{extra space here}
\end{tablenotes}
\end{footnotesize}
\end{threeparttable}
\end{adjustwidth}
\end{table}

\clearpage

\pagebreak

\section{Review of Applications of Instrumental Variables}
\label{app:young}

\vspace{5.4cm}

\setcounter{table}{0}
\renewcommand{\thetable}{\ref{app:young}.\arabic{table}}

\begin{table}[!h]
\begin{adjustwidth}{-1in}{-1in}
\centering
\begin{threeparttable}
\caption{Negative First Stages with Alternative Binarizations\label{tab:apps_neg}}
\begin{tabular}{c|c|>{\centering\arraybackslash}m{1.35cm} >{\centering\arraybackslash}m{1.35cm} c >{\centering\arraybackslash}m{1.35cm} >{\centering\arraybackslash}m{1.35cm} c >{\centering\arraybackslash}m{1.35cm} >{\centering\arraybackslash}m{1.35cm}}
\hline\hline
          &       & \multicolumn{2}{c}{$Z_1$} &       & \multicolumn{2}{c}{$Z_2$} &       & \multicolumn{2}{c}{$Z_3$} \\
\hline
          &       & LPM   & Probit &       & LPM   & Probit &       & LPM   & Probit \\
\hline
    \multirow{2}[0]{*}{$D_1$} & Average Share & 0.293 & 0.250 &       & 0.246 & 0.157 &       & 0.281 & 0.263 \\
          & Weighted Average Share & 0.321 & 0.304 &       & 0.277 & 0.253 &       & 0.315 & 0.327 \\
          &       &       &       &       &       &       &       &       &  \\
    \multirow{2}[0]{*}{$D_2$} & Average Share & 0.272 & 0.220 &       & 0.218 & 0.176 &       & 0.199 & 0.187 \\
          & Weighted Average Share & 0.327 & 0.314 &       & 0.285 & 0.280 &       & 0.263 & 0.259 \\
          &       &       &       &       &       &       &       &       &  \\
    \multirow{2}[0]{*}{$D_3$} & Average Share & 0.379 & 0.341 &       & 0.276 & 0.190 &       & 0.243 & 0.223 \\
          & Weighted Average Share & 0.379 & 0.382 &       & 0.316 & 0.283 &       & 0.293 & 0.307 \\
\hline
\end{tabular}
\begin{footnotesize}
\begin{tablenotes}[flushleft]
\item \textit{Notes:} The table reports summary statistics on the fraction of observations for which $\hat{\e} \left[ D_j \mid Z_k=1, X=x \right] - \hat{\e} \left[ D_j \mid Z_k=0, X=x \right]$ is negative. ``Average Share'' treats every applicable regression equally. ``Weighted Average Share'' weights by the inverse of the number of applicable regressions associated with a given paper. $D_j$ and $Z_k$ are defined as either the original endogenous explanatory variable and instrument (if they are binary) or indicators for whether these variables are above the $j$th and $k$th quartile, respectively, subject to a normalization that $Z_k$ is always associated with a positive estimated coefficient in the linear first stage. Sampling weights and clustered standard errors are used in line with the original papers.
\end{tablenotes}
\end{footnotesize}
\end{threeparttable}
\end{adjustwidth}
\end{table}

\begin{table}[!p]
\begin{adjustwidth}{-1in}{-1in}
\centering
\begin{threeparttable}
\caption{First-Stage Heterogeneity with Alternative Binarizations\label{tab:apps_het}}
\begin{tabular}{c|c|>{\centering\arraybackslash}m{1.35cm} >{\centering\arraybackslash}m{1.35cm} c >{\centering\arraybackslash}m{1.35cm} >{\centering\arraybackslash}m{1.35cm} c >{\centering\arraybackslash}m{1.35cm} >{\centering\arraybackslash}m{1.35cm}}
\hline\hline
          &       & \multicolumn{2}{c}{$Z_1$} &       & \multicolumn{2}{c}{$Z_2$} &       & \multicolumn{2}{c}{$Z_3$} \\
\hline
          &       & LPM   & Probit &       & LPM   & Probit &       & LPM   & Probit \\
\hline
    \multirow{2}[0]{*}{$D_1$} & Rejected Papers & 25/25 & 18/20 &       & 22/25 & 16/18 &       & 25/25 & 15/19 \\
          & Average Share of Rejections & 0.807 & 0.755 &       & 0.781 & 0.765 &       & 0.855 & 0.754 \\
          &       &       &       &       &       &       &       &       &  \\
    \multirow{2}[0]{*}{$D_2$} & Rejected Papers & 24/25 & 19/22 &       & 22/25 & 19/21 &       & 24/25 & 18/21 \\
          & Average Share of Rejections & 0.789 & 0.758 &       & 0.715 & 0.749 &       & 0.789 & 0.768 \\
          &       &       &       &       &       &       &       &       &  \\
    \multirow{2}[0]{*}{$D_3$} & Rejected Papers & 24/25 & 18/20 &       & 24/25 & 18/21 &       & 24/25 & 19/20 \\
          & Average Share of Rejections & 0.824 & 0.747 &       & 0.779 & 0.739 &       & 0.817 & 0.747 \\
\hline
\end{tabular}
\begin{footnotesize}
\begin{tablenotes}[flushleft]
\item \textit{Notes:} The table reports results of Wald tests that the coefficients on the interaction terms in regressions of $D_j$ on $Z_k$, $X$, and $Z_k X$ are jointly equal to zero. ``Rejected Papers'' reports the number of papers for which the Bonferroni $p$-value is less than or equal to 0.05. ``Average Share of Rejections'' reports the average share (across papers) of regressions associated with a given paper for which the corresponding Holm $p$-value is less than or equal to 0.05. $D_j$ and $Z_k$ are defined as either the original endogenous explanatory variable and instrument (if they are binary) or indicators for whether these variables are above the $j$th and $k$th quartile, respectively, subject to a normalization that $Z_k$ is always associated with a positive estimated coefficient in the linear first stage. Sampling weights and clustered standard errors are used in line with the original papers.
\end{tablenotes}
\end{footnotesize}
\end{threeparttable}
\end{adjustwidth}
\end{table}

\begin{table}[!p]
\begin{adjustwidth}{-1in}{-1in}
\centering
\begin{threeparttable}
\caption{Detailed Results on Negative First Stages and First-Stage Heterogeneity\label{tab:apps_1_2}}
\begin{footnotesize}
\begin{tabular}{c>{\centering\arraybackslash}m{2.5cm} >{\centering\arraybackslash}m{2.5cm} >{\centering\arraybackslash}m{0.05cm} >{\centering\arraybackslash}m{1.375cm} >{\centering\arraybackslash}m{1.375cm} >{\centering\arraybackslash}m{0.05cm} >{\centering\arraybackslash}m{1.375cm} >{\centering\arraybackslash}m{1.375cm}}
\hline\hline
          & \multicolumn{2}{c}{\multirow{2}[0]{*}{Average Share of Negative First Stages}} & & \multicolumn{5}{c}{First-Stage Heterogeneity} \\
\cline{5-9}
          & \multicolumn{2}{c}{} &       & \multicolumn{2}{c}{Bonferroni $p$-Value} &       & \multicolumn{2}{c}{Share of Rejections} \\
\cline{2-3}
\cline{5-6}
\cline{8-9}
          & LPM   & Probit &       & LPM   & Probit &       & LPM   & Probit \\
\hline
     \cite{Acemoglu2008} & 0.546 & 0.554 &       & 0     & 0     &       & 1     & 1 \\
     \cite{Albouy2012} & 0.147 & 0.122 &       & 0.115 & 0     &       & 0     & 0.226 \\
     \cite{Alesina2011} & 0.145 & 0.159 &       & 0.001 & 0     &       & 0.064 & 0.011 \\
     \cite{Ananat2011} & 0.020  & 0.223 &       & 0.308 & 1     &       & 0     & 0 \\
     \cite{Autor2013} & 0.009 & 0.036 &       & 0     & 0     &       & 0.894 & 0.894 \\
     \cite{Bazzi2013} & 0.303 & 0.242 &       & 0     & 1     &       & 0.750  & 0 \\
     \cite{Becker2011} & 0.068 & 0.100   &       & 0     & 0     &       & 0.125 & 0.450 \\
     \cite{Bleakley2010} & 0.206 & 0.029 &       & 0     & 0     &       & 1     & 1 \\
     \cite{Brown2012} & 0.134 & 0.143 &       & 0     & 0     &       & 0.375 & 0.375 \\
     \cite{Chalfin2015} & 0.342 & N/A   &       & 0     & 0     &       & 1     & 1 \\
     \cite{ChodorowReich2012} & 0.123 & 0.105 &       & 1     & N/A   &       & 0     & N/A \\
     \cite{Chou2010} & 0.514 & 0.503 &       & 0     & 0     &       & 1     & 1 \\
     \cite{Collins2013} & 0.343 & 0.341 &       & 0     & 0     &       & 1     & 1 \\
     \cite{Decarolis2015} & 0.595 & 0.568 &       & 0     & 0     &       & 0.667 & 1 \\
     \cite{Dinkelman2011} & 0.306 & 0.271 &       & 0     & 0     &       & 0.158 & 1 \\
     \cite{Draca2011} & 0.250  & 0.269 &       & 0     & 0     &       & 1     & 1 \\
     \cite{Guryan2010} & 0.427 & N/A   &       & 0     & N/A   &       & 1     & N/A \\
     \cite{Hornung2014} & 0.116 & 0.151 &       & 0     & 0     &       & 0.909 & 0.889 \\
     \cite{Hunt2010} & 0.245 & 0.335 &       & 0     & 0     &       & 1     & 1 \\
     \cite{James2015} & 0.555 & 0.616 &       & 0     & 0     &       & 1     & 1 \\
     \cite{Kraay2014} & 0.435 & 0.439 &       & 0     & 0     &       & 1     & 1 \\
     \cite{Lipscomb2013} & 0     & 0     &       & 0.018 & 0.002 &       & 1     & 1 \\
     \cite{Moser2014} & 0.671 & 0.626 &       & 0     & N/A   &       & 1     & N/A \\
     \cite{Oreopoulos2006} & 0.463 & 0.434 &       & 0     & 0     &       & 0.941 & 0.875 \\
     \cite{Saiz2011} & 0.152 & 0.165 &       & 0     & N/A   &       & 1     & N/A \\
          &       &       &       &       &       &       &       &  \\
     Number of Regressions & 988   & 930   &       & 988   & 899   &       & 988   & 899 \\
\hline
\end{tabular}
\begin{tablenotes}[flushleft]
\item \textit{Notes:} ``Average Share of Negative First Stages'' reports the average fraction of observations for which $\hat{\e} \left[ D \mid Z=1, X=x \right] - \hat{\e} \left[ D \mid Z=0, X=x \right]$ is negative. ``First-Stage Heterogeneity'' reports results of Wald tests that the coefficients on the interaction terms in regressions of $D$ on $Z$, $X$, and $ZX$ are jointly equal to zero. ``Bonferroni $p$-Value'' reports the product of the smallest $p$-value associated with a given paper and the number of applicable regressions in that paper. ``Share of Rejections'' reports the fraction of applicable regressions associated with a given paper for which the corresponding Holm $p$-value is less than or equal to 0.05. $D$ and $Z$ are defined as either the original endogenous explanatory variable and instrument (if they are binary) or indicators for whether these variables are above their medians, subject to a normalization that $Z$ is always associated with a positive estimated coefficient in the linear first stage. Sampling weights and clustered standard errors are used in line with the original papers.
\end{tablenotes}
\end{footnotesize}
\end{threeparttable}
\end{adjustwidth}
\end{table}

\clearpage

\pagebreak

\section{Reanalysis of \cite{Stevenson2018}}
\label{app:stevenson}

\vspace{5cm}

\setcounter{table}{0}
\renewcommand{\thetable}{\ref{app:stevenson}.\arabic{table}}

\begin{table}[!h]
\begin{adjustwidth}{-1in}{-1in}
\centering
\begin{threeparttable}
\caption{Alternative Estimates of the Effects of Pretrial Detention on Conviction and Incarceration Length\label{tab:stevenson_extra}}
\begin{small}
\begin{tabular}{cc >{\centering\arraybackslash}m{1.7cm} >{\centering\arraybackslash}m{1.7cm} c >{\centering\arraybackslash}m{1.7cm} >{\centering\arraybackslash}m{1.7cm} c >{\centering\arraybackslash}m{1.7cm} >{\centering\arraybackslash}m{1.7cm}}
\hline\hline
          &       & \multicolumn{2}{c}{Specification \#1} &       & \multicolumn{2}{c}{Specification \#2} &       & \multicolumn{2}{c}{Specification \#3} \\
\hline
    A. Effects on Conviction &       & $\hat{\beta}$ & $\hat{\sigma}_{\hat{\beta}}$ &       & $\hat{\beta}$ & $\hat{\sigma}_{\hat{\beta}}$ &       & $\hat{\beta}$ & $\hat{\sigma}_{\hat{\beta}}$ \\
\hline
    MB2SLS &       & 0.1610*** & 0.0360 &       & 0.1860*** & 0.0326 &       & 0.0751*** & 0.0271 \\
    JIVE  &       & 0.2127 & 0.1969 &       & 0.3762 & 0.2783 &       & 0.3251 & 0.3863 \\
          &       &       &       &       &       &       &       &       &  \\
\hline
    B. Effects on Incarceration Length &       & $\hat{\beta}$ & $\hat{\sigma}_{\hat{\beta}}$ &       & $\hat{\beta}$ & $\hat{\sigma}_{\hat{\beta}}$ &       & $\hat{\beta}$ & $\hat{\sigma}_{\hat{\beta}}$ \\
\hline
    MB2SLS &       & 134*** & 47    &       & 140*** & 41    &       & 55    & 43 \\
    JIVE  &       & --313  & 280   &       & --225  & 359   &       & 1,279* & 732 \\
          &       &       &       &       &       &       &       &       &  \\
\hline
    Number of Groups &       & \multicolumn{2}{c}{431} &       & \multicolumn{2}{c}{563} &       & \multicolumn{2}{c}{981} \\
    Number of Observations &       & \multicolumn{2}{c}{327,560} &       & \multicolumn{2}{c}{325,915} &       & \multicolumn{2}{c}{319,573} \\
\hline
\end{tabular}
\end{small}
\begin{footnotesize}
\begin{tablenotes}[flushleft]
\item \textit{Notes:} The data are \cite{Stevenson2018}'s sample of 331,971 arrests in Philadelphia. The outcomes are conviction (Panel A) or incarceration length (Panel B), defined as the maximum days of an incarceration sentence. The treatment is pretrial detention. The instrument is whether a given case was heard by Judge C\@. Each specification is based on a division of the sample into a number of mutually exclusive groups, with a separate group for each combination of values of selected variables. Specification \#1 uses the offense type and race (Black, White, or other) of the defendant. Specification \#2 uses the offense type, race, and gender (male or female) of the defendant. Specification \#3 uses the offense type, race and gender of the defendant, and three time periods considered by \cite{Stevenson2018}. Groups with fewer than three observations in either $(G,Z)$ combination are dropped. MB2SLS and JIVE are based on the interacted specification and are described in Section \ref{sec:manyiv}.
\item *Statistically different from zero at the 10\% level; **at the 5\% level; ***at the 1\% level.
\end{tablenotes}
\end{footnotesize}
\end{threeparttable}
\end{adjustwidth}
\end{table}

\begin{table}[!p]
\begin{adjustwidth}{-1in}{-1in}
\centering
\begin{threeparttable}
\caption{Bootstrap $p$-Values for the Comparison of the Noninteracted and Interacted Specifications\label{tab:stevenson_boot1}}
\begin{small}
\begin{tabular}{cc >{\centering\arraybackslash}m{3.5cm} c >{\centering\arraybackslash}m{3.5cm} c >{\centering\arraybackslash}m{3.5cm}}
\hline\hline
          &       & Specification \#1 &       & Specification \#2 &       & Specification \#3 \\
\hline
    A. Effects on Conviction &       & \\
\hline
    2SLS  &       & 0.496 &       & 0.486 &       & 0.301 \\
    MB2SLS &       & 0.829 &       & 0.956 &       & 0.377 \\
    JIVE  &       & 0.818 &       & 0.119 &       & 0.195 \\
    UJIVE &       & 0.619 &       & 0.671 &       & 0.229 \\
          &       &       &       &       &       &  \\
\hline
    B. Effects on Incarceration Length &       & \\
\hline
    2SLS  &       & 0.051 &       & 0.019 &       & 0.030 \\
    MB2SLS &       & 0.049 &       & 0.021 &       & 0.017 \\
    JIVE  &       & 0.001 &       & 0.001 &       & 0.031 \\
    UJIVE &       & 0.032 &       & 0.016 &       & 0.016 \\
          &       &       &       &       &       &  \\
\hline
    Number of Groups &       & 431 &       & 563 &       & 981 \\
    Number of Observations &       & 327,560 &       & 325,915 &       & 319,573 \\
\hline
\end{tabular}
\end{small}
\begin{footnotesize}
\begin{tablenotes}[flushleft]
\item \textit{Notes:} This table revisits the estimates in Tables \ref{tab:stevenson_3} and \ref{tab:stevenson_extra}. The data are \cite{Stevenson2018}'s sample of 331,971 arrests in Philadelphia. The outcomes are conviction (Panel A) or incarceration length (Panel B), defined as the maximum days of an incarceration sentence. The treatment is pretrial detention. The instrument is whether a given case was heard by Judge C\@. Each specification is based on a division of the sample into a number of mutually exclusive groups, with a separate group for each combination of values of selected variables. Specification \#1 uses the offense type and race (Black, White, or other) of the defendant. Specification \#2 uses the offense type, race, and gender (male or female) of the defendant. Specification \#3 uses the offense type, race and gender of the defendant, and three time periods considered by \cite{Stevenson2018}. Groups with fewer than three observations in either $(G,Z)$ combination are dropped. Each $p$-value is calculated using a bootstrap test of the equality of the estimands in the noninteracted and interacted specifications (with 250 bootstrap replications). The noninteracted specification is estimated using IV\@. The estimators of the interacted specification, listed in the first column, are described in Section \ref{sec:manyiv}.
\end{tablenotes}
\end{footnotesize}
\end{threeparttable}
\end{adjustwidth}
\end{table}

\begin{table}[!p]
\begin{adjustwidth}{-1in}{-1in}
\centering
\begin{threeparttable}
\caption{Bootstrap $p$-Values for the Comparison of 2SLS and Other Estimators of the Interacted Specification\label{tab:stevenson_boot2}}
\begin{small}
\begin{tabular}{cc >{\centering\arraybackslash}m{3.5cm} c >{\centering\arraybackslash}m{3.5cm} c >{\centering\arraybackslash}m{3.5cm}}
\hline\hline
          &       & Specification \#1 &       & Specification \#2 &       & Specification \#3 \\
\hline
    A. Effects on Conviction &       & \\
\hline
    MB2SLS &       & 0.008 &       & 0.000 &       & 0.450 \\
    JIVE  &       & 0.018 &       & 0.000 &       & 0.000 \\
    UJIVE &       & 0.188 &       & 0.057 &       & 0.176 \\
          &       &       &       &       &       &  \\
\hline
    B. Effects on Incarceration Length &       & \\
\hline
    MB2SLS &       & 0.404 &       & 0.490 &       & 0.030 \\
    JIVE  &       & 0.000 &       & 0.000 &       & 0.000 \\
    UJIVE &       & 0.010 &       & 0.108 &       & 0.006 \\
          &       &       &       &       &       &  \\
\hline
    Number of Groups &       & 431 &       & 563 &       & 981 \\
    Number of Observations &       & 327,560 &       & 325,915 &       & 319,573 \\
\hline
\end{tabular}
\end{small}
\begin{footnotesize}
\begin{tablenotes}[flushleft]
\item \textit{Notes:} This table revisits the estimates in Tables \ref{tab:stevenson_3} and \ref{tab:stevenson_extra}. The data are \cite{Stevenson2018}'s sample of 331,971 arrests in Philadelphia. The outcomes are conviction (Panel A) or incarceration length (Panel B), defined as the maximum days of an incarceration sentence. The treatment is pretrial detention. The instrument is whether a given case was heard by Judge C\@. Each specification is based on a division of the sample into a number of mutually exclusive groups, with a separate group for each combination of values of selected variables. Specification \#1 uses the offense type and race (Black, White, or other) of the defendant. Specification \#2 uses the offense type, race, and gender (male or female) of the defendant. Specification \#3 uses the offense type, race and gender of the defendant, and three time periods considered by \cite{Stevenson2018}. Groups with fewer than three observations in either $(G,Z)$ combination are dropped. Each $p$-value is calculated using a bootstrap test of the equality of the probability limits of 2SLS and other estimators in the interacted specifications (with 250 bootstrap replications). Other estimators of the interacted specification, listed in the first column, are described in Section \ref{sec:manyiv}.
\end{tablenotes}
\end{footnotesize}
\end{threeparttable}
\end{adjustwidth}
\end{table}

\clearpage

\end{appendices}

\end{document}